\documentclass[11pt]{article}
\usepackage[margin=1in,bottom=.8in]{geometry} 

\usepackage[nocompress]{cite}

\usepackage{amsmath}
\usepackage{amssymb}
\usepackage{mathtools}
\usepackage{mathrsfs}    
\usepackage{caption}
\usepackage[normalem]{ulem}

\usepackage{amsthm}
\usepackage{thmtools,thm-restate}  
\usepackage{scrextend}

\usepackage{mathpartir}
\usepackage{stmaryrd}
\usepackage{proof}

\usepackage[colorlinks=true, urlcolor=purple, linkcolor=blue]{hyperref}
\usepackage[dvipsnames]{xcolor}
\usepackage{xspace}                

\newcommand{\mc}{\mathcal}

\newcounter{itemnum}

\newcommand{\case}[2][]{\vskip12pt\noindent{#2:}\phantomsection\label{#1}}
\newcommand{\refcase}[2][]{\vskip12pt\noindent{\hyperref[#1]{#2:}}}
\newcommand{\subcase}[1]{\vskip12pt\indent{#1:}}
\newcommand{\textlabel}[2]{{\color{PineGreen}
    {\refstepcounter{itemnum}(\theitemnum)\,#1\label{#2}}}}
\newcommand{\resetitemnum}{\setcounter{itemnum}{0}}

\newcommand{\st}{%
  \nonscript\;
  \ifnum\currentgrouptype=16
    \middle\vert
  \else
    \mid
  \fi
  \nonscript\;}

\newcommand{\imp}{\supset}

\renewcommand{\subset}{\subseteq}

\newcommand{\ind}{\stackrel{\mathclap{\scalebox{0.66}{\mbox{$\mu$}}}}{=}}
\newcommand{\defn}{\stackrel{\mathclap{\scalebox{0.5}{\mbox{$\Delta$}}}}{=}}
\newcommand{\defstar}{\stackrel{\mathclap{\scalebox{0.66}{\mbox{$*$}}}}{=}}

\newcommand{\dfn}{\mathrm{defn}}

\newcommand{\R}{\mathscr{R}}

\newcommand{\IR}{$\mu\mc R$\xspace}
\newcommand{\IL}{$\mu\mc L$\xspace}
\newcommand{\muR}{\IR}
\newcommand{\muL}{\IL}

\newcommand{\DR}{$\Delta\mc R$\xspace}
\newcommand{\DL}{$\Delta\mc L$\xspace}
\newcommand{\axiom}{\textsc{Ax}\xspace}
\newcommand{\init}{\textsc{Init}\xspace}
\newcommand{\cL}{c$\mc L$\xspace}
\newcommand{\wL}{w$\mc L$\xspace}

\newcommand{\multc}{$\textit{mc}$\xspace}
\newcommand{\topR}{$\top\mc R$\xspace}
\newcommand{\botL}{$\bot\mc L$\xspace}
\newcommand{\impR}{$\imp\!\!\mc R$\xspace}
\newcommand{\impL}{$\imp\!\!\mc L$\xspace}
\newcommand{\andR}{$\wedge\mc R$\xspace}
\newcommand{\andL}[1]{$\wedge\mc L#1$\xspace}
\newcommand{\orR}[1]{$\vee\mc R#1$\xspace}
\newcommand{\orL}{$\vee\mc L$\xspace}
\newcommand{\allR}{$\forall\mc R$\xspace}
\newcommand{\allL}{$\forall\mc L$\xspace}
\newcommand{\exR}{$\exists\mc R$\xspace}
\newcommand{\exL}{$\exists\mc L$\xspace}
\newcommand{\nabR}{$\nabla\mc R$\xspace}
\newcommand{\nabL}{$\nabla\mc L$\xspace}

\newcommand{\G}{$\mc G$\xspace}
\newcommand{\LD}{$\text{LD}$\xspace}
\newcommand{\LDinf}{$\text{LD}_\infty$\xspace}
\newcommand{\LDIN}{$\text{LD}^{\mu\nabla}$\xspace}
\newcommand{\LDINinf}{$\text{LD}^{\mu\nabla}_\infty$\xspace}

\newcommand{\NJI}{$\mu\text{NJ}$\xspace}

\newcommand{\multcut}{mc}
\newcommand{\weakL}{\text{\wL}}
\newcommand{\forallR}{\text{\allR}}
\newcommand{\forallL}{\text{\allL}}
\newcommand{\existsL}{\text{\exL}}
\newcommand{\existsR}{\text{\exR}}
\newcommand{\circL}{\circ_{L}}
\newcommand{\circR}{\circ{\cal R}}
\newcommand{\bulletL}{\bullet_{L}}

\newcommand{\deriv}[2]{{\begin{array}{c}#1\\#2\end{array}}}

\newcommand{\Seq}[2]{#1\vdash #2}

\newcommand{\Perm}{\mathrm{Perm}}
\newcommand{\supp}{\mathrm{supp}}

\newcommand{\lvl}{\mathrm{lvl}}
\newcommand{\ground}{\mathrm{ground}}

\newcommand{\height}{\mathrm{ht}}
\newcommand{\Ind}{\mathrm{Ind}}

\newcommand{\range}{\mathrm{range}}

\newtheorem{definition}{Definition}
\newtheorem{theorem}{Theorem}
\newtheorem{lemma}{Lemma}
\newtheorem{sublemma}[lemma]{Sublemma}
\newtheorem{corollary}{Corollary}

\title{Towards weak stratification for logics of definitions}
\author{Nathan Guermond}
\date{April 17, 2025\footnote{Last modified: \today}}

\begin{document}
\maketitle

\begin{abstract}
  The logic of definitions is a family of logics for encoding and reasoning
  about judgments, which are atomic predicates specified by inference rules. 
A definition associates an atomic predicate with a logical formula, which may itself 
depend on the predicate being defined. This leads to an apparent circularity
which can be resolved by interpreting definitions as monotone fixed-point
operators on terms, and which is enforced by imposing a stratification
condition on definitions. In many instances, it is useful to consider definitions
in which the predicate being defined appears negatively in the body of its
definition. In the logic \G, underlying the Abella proof assistant, this is not allowed due to the stratification condition.
One such application violating this condition is that of
defining logical relations, which is a technique commonly used
to prove properties about programming languages. Tiu has shown how
to relax this stratification condition to allow for a broader body
of definitions including that needed for logical relations. However,
he only showed how to extend a core fragment of \G with the weakened stratification condition, resulting in a logic he called \LD.
In this work we show that the weakened stratification
condition is also compatible with generic (nabla) quantification and general
induction. The eventual aim of this work is to justify an extension
of the Abella proof assistant allowing for such definitions.
\end{abstract}


\section{Introduction}
{\newcommand{\nil}{nil}
\newcommand{\cons}{::}
\newcommand{\listp}{list}
\newcommand{\append}{append}
\newcommand{\nat}{nat}
\newcommand{\app}{app}

The logic of definitions is a family of logics for encoding and reasoning
about judgments, which are atomic predicates specified by inference rules. 
A definition associates an atomic predicate to a logical formula, which may itself 
depend on the predicate being defined. 

For example, consider the following specification for constructing lists and appending them:

\begin{mathpar}
  \inferrule*[right=nil]
  { }{\listp\ \nil}\and
  \inferrule*[right=cons]
  {\listp\ L}{\listp\ (X :: L)}\\
  \inferrule*[right=AppNil]
  { }{\append\ \nil\ K\ K}\and
  \inferrule*[right=AppCons]
  {\append\ L\ K\ M}
  {\append\ (X :: L)\ K\ (X :: M)}
\end{mathpar}
\begin{figure}[h!]
  \begin{mathpar}
  \listp\ nil\defn\top\and
  \listp\ (X :: L) \defn \listp\ L\\
  \append\ \nil\ K\ K \defn \top\and
  \append\ (X :: L)\ K\ (X :: M)\defn \append\ L\ K\ M
\end{mathpar}
\caption{Definitions for $\listp$ and $\append$}
\label{fig:append}
\end{figure}
These inference rules may be modeled in the logic of definitions with the definitions for $\listp$ and $\append$ given by the clauses in Figure~\ref{fig:append}. Each clause models an inference rule with the conclusion appearing as an atomic formula on the left (the head), and the premises appearing as a formula on the right (the body). Variables (capitalized), which may appear in both the head and body, are used to model rule schemas whose variables may be instantiated for any term of the same type.
The derivability of a judgment under some hypotheses may then be modeled as the provability of a formula, as in the following example
$$\forall X,Y,L,K,M.\append\ L\ K\ M \imp  \append\ (X :: Y :: L)\ K\ (X :: Y :: M),$$
Provability in the logic is thus defined in such a way that proving the above models the derivability of the judgment $\append\ (x :: y :: \ell)\ k\ (x :: y :: m)$ for closed terms $x,y,\ell,k,m$ for which $\append\ \ell\ k\ m$ is known.
Conversely, non-derivability of a judgment may be modeled by the provability of the negation of the corresponding formula. For example, proving that
$$\forall X.\append\ \nil\ (X ::\nil)\ \nil \imp \bot$$
means that the judgment $\append\ \nil\ (x ::\nil)\ \nil$ is not derivable for any closed term $x$. Indeed, this should be so because there are no inference rules whose conclusion matches with this judgment.

In programming languages research, one often wants to discuss and reason about rules-based specifications, such as typing systems, logics, or operational semantics for programming languages (see eg. \cite{Ha16, GTL89}).
In a similar manner to the above examples, the logic of definitions may be used to model and reason about such specifications.

When specifying definitions, one notices that the body of a definition is allowed to depend on the predicate being defined, which may lead to an apparent circularity. Indeed, failing to impose restrictions on the types of definitions we are allowed to make can lead to inconsistencies. For example, allowing the definition $p \defn p\imp \bot$ leads to an inconsistent logic, as we will see in Section~\ref{section:definitions}. The logic of definitions therefore imposes a \emph{stratification} condition on definitions, which prevents the predicate being defined from appearing negatively (ie. to the left of an implication) in the body of its definition. Alternatively, one can think of definitions as fixed-points of operators on terms, where the stratification condition corresponds to ensuring that these operators are monotone. For example, a corresponding fixed-point operator for the $\listp$ predicate in Figure~\ref{fig:append} is
  \begin{equation}
    \label{eq:lfp-operator}
    B \equiv \lambda p. \lambda L. (L = \nil)\vee (\exists X,L'. (L = X :: L')\wedge (p\ L')
  \end{equation}
  where, instead of specifying a definition by a set of clauses, we could instead specify a definition as a fixed-point of an operator like $B$ above. This is the approach we will take for \emph{inductive} definitions later, which will allow us to think of inductive definitions as \emph{least} fixed-points.

In many instances, however, it is useful to consider definitions in which the predicate being defined appears negatively in the body of its definition. One such application is that of defining logical relations, which is a technique commonly used to prove properties about programming languages (see eg. \cite{Ha16,GTL89,WN16}). For example, \emph{logical equivalence}\footnote{As opposed to computational equivalence of programs, which relates programs intensionally (ie. syntactically), logical equivalence relates programs extensionally (ie. semantically), so $\lambda x.\lambda y. x + y$ and $\lambda x\lambda y. y + x$ are logically, but not computationally, equivalent.} of programs in the simply typed lambda calculus is defined by recursion on types in Figure~\ref{fig:logeq},
\begin{figure}[h]
\begin{align*}
  s \sim_{\nat} t &\defn \exists n. (s\to^* n) \wedge (t\to^* n) \wedge (\nat\ n)\\
  s \sim_{\sigma\to\tau} t &\defn \forall u_1,u_2. (u_1\sim_\sigma u_2)\imp (\app\ s\ u_1\sim_\tau \app\ t\ u_2)
\end{align*}
\caption{Logical equivalence in the simply typed lambda calculus}
\label{fig:logeq}
\end{figure}
where we assume $\to^*$ is defined as the reduction relation on lambda terms, and $\nat$ is defined to hold on natural numbers.
This definition is not stratified because $u_1\sim_\sigma u_2$ appears negatively in the body of the definition. However, if we think of fixing a predicate $\sim_\tau$ for each ground type $\tau$, where we think of $\sim_{\sigma\to\tau}$ as being defined only once $\sim_\sigma$ has been defined, then we see that this family of predicates can still be stratified.

In \cite{Tiu12}, Tiu showed how to make this intuition precise by weakening the stratification condition by allowing stratification to depend on ground arguments. Tiu was then able to show that his stratification condition, which we will refer to as \emph{weak stratification}, is compatible with a core fragment of the logic \G. This logic, which he called \LD, consists of first order logic extended with weakly stratified definitions and natural number induction.

In this work, we take this idea further by showing that the weakened stratification condition is also compatible with the nabla quantifier and inductive definitions in a logic we call \LDIN. The eventual aim of this work is to justify an extension of the Abella proof assistant \cite{Abella14} allowing for weakly stratified definitions. Since Abella is based on the logic \G introduced in \cite{GMN11}, fully justifying this extension would require further extending \LDIN with co-induction and nominal abstraction.

The nabla quantifier ($\nabla$), first introduced by Miller and Tiu in \cite{MT05} and refined in \cite{Tiu07}, was introduced in order to reason about names, which are modeled by \emph{nominal constants}, and freshness, which is modeled by the nabla quantifier. For a judgment $\nabla x.J$ to be derivable means that $J[n/x]$ has a derivation $\Pi[n/x]$ for a generic nominal constant $n$ not appearing in $J$. In particular, this allows us to conclude that $J[t/x]$ is uniformly derivable by the derivation $\Pi[t/x]$ for any term $t$.  This differs from the meaning of a universally quantified judgment $\forall x.J$, which asserts that $J[t/x]$ has a derivation $\Pi_t$ for every term $t$, but cannot guarantee the uniformity of the derivations.


In a similar vein, inductive definitions were introduced by Momigliano and Tiu in \cite{MT04} to allow us to reason about the least fixed-point of a definition. For example, to model the least fixed-point $p$ of an operator $B$ (eg. the operator for $\listp$ in (\ref{eq:lfp-operator})), it should be the case that for any predicate $S$ such that
$$\forall X. B\ S\ X\imp S\ X$$
is provable (ie. $S$ is a pre-fixed-point  of $B$), and for which $S\ t\imp C$ is provable, then $p\ t \imp C$ should be provable. In \LDIN we will introduce a general rule capturing this idea. This is in contrast to Tiu's \LD, which supports natural number induction, but does not support induction in this general sense.

In order for a logic to be useful, we must first show that it is \emph{consistent}, which means that not all formulas are provable. To show consistency of \LD, Tiu first constructs a ground version of the logic, \LDinf, which means that no eigenvariables (free variables) appear in the derivations. He then defines an interpretation of \LD into \LDinf, which shows that if \LDinf is inconsistent, then \LD must be as well. This reduces the problem to proving the consistency of \LDinf which he does by using the cut elimination technique. The aim of this method is to show that any derivation making use of the \emph{cut} rule can be transformed into a \emph{cut-free} derivation, which allow us to locally decide the provability of certain formulas, including $\bot$. This is done by defining a \emph{cut reduction} relation, which can be shown to be well founded. The technique for doing this is known as the Tait-Martin-L\"of reducibility technique \cite{Ta67,ML71}, which was adapted to the logic of definitions by McDowell and Miller \cite{MDM00}. Following Tiu's work, we define a ground version of \LDIN called \LDINinf and prove cut-elimination for the latter, thus allowing us to conclude that \LDIN is consistent.

Our main contribution is to show that \LD extended
with nabla quantification and inductive definitions is consistent,
which we show by adapting Tiu's method of eliminating cut in a ground logic.
Doing this required resolving
weak stratification with inductive definitions using the ideas
of Momigliano and Tiu~\cite{MT04}, which assume strict stratification
on both generic and inductive definitions. Our results show that this
can be done so long as the strict stratification condition is
maintained on inductive definitions. We show by means of a counterexample
that allowing weakly stratified inductive definitions may lead to an
inconsistency.
  
The rest of this paper is structured as follows. In section~\ref{sec:ldin}, we introduce the logic \LDIN, which we wish to prove consistent. In section~\ref{sec:examples}, we give examples of definitions requiring weak stratification, as well as a counter example justifying the restriction to strict stratification for inductive definitions. In section~\ref{sec:ldininf} we introduce a ground version of the logic, \LDINinf. We show in section~\ref{sec:interp} that \LDIN may be interpreted in \LDINinf, allowing us to conclude that the consistency of \LDIN may be reduced to the consistency of \LDINinf. In section~\ref{sec:ldininf-cutred}, we define the cut reduction relation on ground derivations, and in section~\ref{sec:cutelim} we introduce the intermediate notion of reducibility and prove that all ground derivations are reducible. We put these results together to prove the cut elimination theorem in~\ref{sec:cut-elimination}. In section~\ref{sec:applications} we show how to apply this result to obtain consistency of \LDIN, as well as an application to proof search. We conclude in section~\ref{sec:conclusion} with a discussion of related work and future directions.
 
}

\section{The logic \LDIN}
\label{sec:ldin}
{
\newcommand{\nat}{\textrm{nat}}
\newcommand{\nil}{nil}
\newcommand{\cons}{::}
\newcommand{\listp}{list}
\newcommand{\append}{append}
\newcommand{\lst}{\textrm{lst}}
\newcommand{\eq}{eq}

In this section we define the logic \LDIN, which is an extension of the first order sequent calculus implicitly parametrized by a set of definitions $\mc D$ and a set of inductive definitions $\mc I$. First, we define the syntax of terms and formulas in the logic, and what constitutes a valid sequent. Second, we describe a core fragment of the logic, consisting of logical inference rules, followed by the structural, axiom, and multicut rules. Then, we specify what constitutes a valid set of definitions $\mc D$, as well as the inference rules associated with $\mc D$. Finally, we specify what constitutes a valid set of inductive definitions $\mc I$, with its associated inference rules. The logic \LDIN, which is implicitly parametrized by the sets $\mc D$ and $\mc I$, is precisely summed up at the end of this section. We observe that \LDIN is an extension of Tiu's \LD with nabla quantification and inductive definitions. It can also be viewed as a core fragment of the logic $\mc G$ \cite{GMN11}, which underlies Abella, extended with more permissive definitions.

\subsection{Syntax}

As in Church's simple theory of types \cite{Ch40}, we assume a finite set of base types $\iota_1,\ldots,\iota_n$, arrow types $\alpha\to\beta$, and a distinguished type $o$ of propositions. We assume we are given a \emph{signature} $\Sigma$, which is a set of type annotated constants of the form $\kappa : \tau$. Following the presentation in \cite{GMN08,Tiu07}, we assume a countably infinite number of global \emph{nominal constants} $n^\tau_1,n^\tau_2,\ldots$ in each type $\tau$, but we will omit the typing annotations for simplicity. Terms in our logic are then well-typed terms in the simply typed lambda calculus modulo $\alpha,\beta,\eta$-equality constructed from constants in the signature $\Sigma$, nominal constants, and variables from a \emph{variable context} $\mc X$, which is a finite set of distinct type annotated variables $x : \tau$. Since a term $t$ only exists with respect to a signature $\Sigma$, nominal constants, and a variable context $\mc X$, and since $\Sigma$ and the set of nominal constants are fixed, we will say that $t$ is well formed with respect to $\mc X$, or that $t$ \emph{lies over} $\mc X$. In particular this means that $\mc X$ must contain all free variables in $t$, 
but may contain more. In addition we remark that the type of $t$ is entirely determined by $\Sigma$, $\mc X$, and the nominal constants, and thus we denote $\mc X \vdash_\Sigma t : \tau$ whenever $t$ lies over $\mc X$.

We call a type \emph{first-order} if it does not contain $o$, and \emph{propositional} if it is either $o$ or of the form $\tau\to \omega$ for a first order type $\tau$ and a propositional type $\omega$. A \emph{formula} is a term of type $o$. We assume that our signature $\Sigma$ contains the following logical constructors
\begin{mathpar}
  \bot,\top : o\and
  \wedge,\vee,\imp\ : o\to o\to o\and
  \forall_\alpha,\exists_\alpha,\nabla_\alpha : (\alpha\to o)\to o
\end{mathpar}
for every first-order type $\alpha$. \emph{Atomic} formulas are of the form $p\ \vec t$ for some constant predicate symbol $p : \omega$ in $\Sigma$ where in general, $u\ \vec t$ is an abbreviation for the application $(\ldots(u\ t_1)\ \ldots\ t_n)$. We abbreviate $Q_\alpha(\lambda x.C)$ (where $Q_\alpha = \exists_\alpha,\forall_\alpha,\nabla_\alpha$) by $Q_\alpha x.C$.

We define the \emph{support} of a formula $F$ to be the set of nominals that appear in $F$, which we denote $\supp(F)$. Note that the support of a formula is necessarily finite. Given a finite set of nominals $\mc N$, we may consider the set $\Perm(\mc N)$ of all permutations of these nominals preserving their type. Given a permutation $\pi\in\Perm(\mc N)$ such that $\supp(F)\subset \mc N$, let $F[\pi]$ denote the formula $F$ with each nominal renamed by $\pi$.

A \emph{sequent} is a judgment $\mc X; \Gamma \vdash C$ for a finite multiset of formulas $\Gamma$, called a \emph{context}, and a formula $C$ lying over a variable context $\mc X$. We call the variables in $\mc X$ \emph{eigenvariables}, and we will assume that all the types in $\mc X$ are first-order. We will omit the type annotations when they can be inferred from the context. A term is \emph{ground} if it lies over the empty variable context $\emptyset$. We denote the set of ground terms of type $\alpha$ by $\ground(\alpha)$. Given variable contexts $\mc X$ and $\mc Y$, a simultaneous substitution $\theta = [t_1/x_1,\ldots,t_n/x_n]$ of type $\mc Y\to\mc X$ is an assignment of terms $t_1 : \tau_1,\ldots,t_n:\tau_n$ lying over $\mc Y$ to each variable $x_1:\tau_1,\ldots,x_n:\tau_n = \mc X$. So given a term $t$ lying over $\mc X$, $t\theta := (\lambda x_1\ldots \lambda x_n. t)\ \vec t$ is a term lying over $\mc Y$. Moreover, we write $[t/x]$ to denote the simultaneous substitution $[x_1/x_1,\ldots,x_n/x_n,t/x] : \mc X\to \mc X,x$, and we denote the trivial substitution by $\epsilon_{\mc X}:\mc X\to\mc X$. Finally, the \emph{range} of a substitution, which we denote $\range(\theta)$, is the smallest variable context containing all the free variables appearing in it.

\subsection{Logical rules}

\begin{figure}[h!]
  \begin{mathpar}
    {\inferrule*[right=\botL]{ }{\mc X;\Gamma,\bot\vdash B}}\and
    {\inferrule*[right=\topR]{ }{\mc X;\Gamma\vdash \top}}\\
    {\inferrule*[right=\impL]{\mc X;\Gamma\vdash B\quad \mc X;\Gamma,C\vdash D}
      {\mc X;\Gamma, B\imp C\vdash D}}\and
    {\inferrule*[right=\impR]{\mc X;\Gamma,B\vdash C}{\mc X;\Gamma\vdash B\imp C}}\\
    {\inferrule*[right=\andL1]{\mc X;\Gamma,B\vdash D}{\mc X;\Gamma,B\wedge C\vdash D}}\and
    {\inferrule*[right=\andL2]{\mc X;\Gamma,C\vdash D}{\mc X;\Gamma,B\wedge C\vdash D}}\and
    {\inferrule*[right=\andR]{\mc X;\Gamma\vdash B\quad \mc X;\Gamma\vdash C}
      {\mc X;\Gamma\vdash B\wedge C}}\\
    {\inferrule*[right=\orL]{\mc X;\Gamma,B\vdash D\quad\mc X;\Gamma,C\vdash D}
      {\mc X;\Gamma,B\vee C\vdash D}}\and
    {\inferrule*[right=\orR1]{\mc X;\Gamma\vdash B}{\mc X;\Gamma\vdash B\vee C}}\and
    {\inferrule*[right=\orR2]{\mc X;\Gamma\vdash C}{\mc X;\Gamma\vdash B\vee C}}\\
    {\inferrule*[right=\allL]
      {\mc X;\Gamma,C[t/x]\vdash D\\
        \mc X \vdash_\Sigma t : \tau}
      {\mc X;\Gamma,\forall_\tau x.C\vdash D}}\and
    {\inferrule*[right={\allR\quad\raisebox{.5em}{\parbox[t]{6em}
          {$\supp(C) = \vec n$\\ $y\not\in \mc X$}}}]
      {\mc X, y ; \Gamma\vdash C[y\ \vec n/x]}
      {\mc X ;\Gamma\vdash \forall_\tau x.C}}\\
    {\inferrule*[right={\exL\quad\raisebox{.5em}{\parbox[t]{6em}
          {$\supp(C) = \vec n$\\$y\not\in\mc X$}}}]
      {\mc X,y;\Gamma,C[y\ \vec n/x]\vdash D}
      {\mc X;\Gamma,\exists_\tau x.C\vdash D}}\and
        {\inferrule*[right=\exR]
          {\mc X;\Gamma\vdash C[t/x]\\
          \mc X\vdash_\Sigma t : \tau}
      {\mc X;\Gamma\vdash \exists_\tau x.C}}\\
    {\inferrule*[right={\nabL\quad$n\not\in\supp(C)$}]
      {\mc X;\Gamma,C[n/x]\vdash D}
      {\mc X;\Gamma,\nabla_\tau x.C\vdash D}}\and
    {\inferrule*[right={\nabR\quad$n\not\in\supp(C)$}]
      {\mc X;\Gamma\vdash C[n/x]}
      {\mc X;\Gamma\vdash \nabla_\tau x.C}}
  \end{mathpar}
  \caption{Core inference rules for \LDIN}
  \label{fig:core-rules}
\end{figure}

We now present the core left and right introductions rules for \LDIN in Figure~\ref{fig:core-rules}. We say that a sequent $\mc X;\Gamma\vdash C$ is \emph{derivable} if there exists a rule whose conclusion matches this sequent, and for which each of that rule's premises is derivable.

Note in particular the the right universal (respectively, left existential) introduction rules where $C[y\ \vec n/x]$ for $\vec n = \supp(C)$. This differs from the usual presentation of first order logic, and ensures the non-derivability of the following sequent
$$\mc X;\forall y.\nabla x.p\ x\ y\vdash \nabla x.\forall y.p\ x\ y$$
for a given predicate constant $p$. The application $(y\ \vec n)$ keeps track of the scope in which $y$ was introduced, a technique known as \emph{raising} (see \cite{M92}). We note in particular that the type of $y$ above is $\tau_1\to\ldots\to \tau_n\to\tau$ whenever $n_i : \tau_i$. 


\subsection{Structural, axiom, and multicut rules}

\begin{figure}[h!]
  \begin{mathpar}
    {\inferrule*[right=\cL]{\mc X;\Gamma,B,B\vdash C}{\mc X;\Gamma,B\vdash C}}\and
    {\inferrule*[right=\wL]{\mc X;\Gamma\vdash C}{\mc X;\Gamma,B\vdash C}}\\
    {\inferrule*[right={\axiom\quad$\pi\in\Perm(\supp(A))\ \mathrm{and}\ A\ \mathrm{atomic}$}]{ }{\mc X;A\vdash A[\pi]}}\and
    {\inferrule*[right=\multc]{\mc X;\Delta_1\vdash A_1\qquad\ldots\qquad\mc X;\Delta_n\vdash A_n\qquad \mc X;\Gamma,A_1,\ldots,A_n\vdash C}{\mc X;\Gamma,\Delta_1,\ldots,\Delta_n\vdash C}}
  \end{mathpar}
  \caption{Structural, axiom, and multicut rules for \LDIN}
  \label{fig:struct-rules}
\end{figure}

We then give the structural contraction and weakening (resp. \cL and \wL), axiom (\axiom), and multicut (\multc) rules in Figure~\ref{fig:struct-rules}. The contraction rule is necessary because contexts are multisets, and the weakening rule is necessary because the axiom rule does not permit multiple premises. Finally, the cut rule is presented as a multicut rule, an idea due to Gentzen which will allow us to more easily reduce a cut involving the contraction rule in Section~\ref{sec:ldininf-cutred}. We also permit the multicut rule to be nullary (ie. $n \geq 0$ in Figure~\ref{fig:struct-rules}).

Note that the axiom rule is restricted to atomic formulas. The following lemma shows this is not a fundamental restriction, but using this rule will make the technical developments involving inductive definitions easier.

\begin{lemma}
For any permutation $\pi\in\supp(B)$, the following \init rule is derivable without cut:
\begin{mathpar}
{\inferrule*[right={\init\quad$\pi\in\Perm(\supp(B))$}]{ }{\mc X;B\vdash B[\pi]}}
\end{mathpar}
\end{lemma}
\begin{proof} By induction on the structure of the formula $B$. We only show the case in which $B = B_1\wedge B_2$. In this case, we know by the inductive hypothesis that $\mc X;B_1\vdash B_1[\pi]$ and $\mc X;B_2\vdash B_2[\pi]$, from which we may form the derivation
  \begin{mathpar}
    \inferrule*[right=\andR]
    {\inferrule*[right=\andL1]
      {\mc X;B_1\vdash B_1[\pi]}
      {\mc X;B_1\wedge B_2\vdash B_1[\pi]}\\
      \inferrule*[Right=\andL2]
      {\mc X;B_2\vdash B_2[\pi]}
      {\mc X;B_1\wedge B_2\vdash B_2[\pi]}}
    {\mc X;B_1\wedge B_2 \vdash (B_1\wedge B_2)[\pi]}
  \end{mathpar}
\end{proof}

\subsection{Definitions and stratification}
\label{section:definitions}
A valid set of definitions $\mc D$ consists of clauses of the form ${p\ \vec t \defn_{\mc X} B \in \mc D}$ for some propositional constant $p : \omega$ in the signature $\Sigma$ and a formula $B$ with $\vec t$ and $B$ both lying over $\mc X$ and having empty support. For a given clause $H\defn_{\mc X} B\in\mc D$, we say that $H$ is the \emph{head} of the clause, and $B$ the \emph{body}. We furthermore require that for any clause $H\defn_{\mc X}B\in\mc D$, every variable in $\mc X$ appears in $H$, and that $H$ lies in the $L_\lambda$ \emph{pattern fragment} (see eg. \cite{MN12}). For a given predicate constant $p$, the \emph{definition} associated with $p$ is the set of clauses in $\mc D$ with $p$ in the head. We now assume that for every ground atomic formula $A$ we have a \emph{level} assignment $\lvl(A)$ to an ordinal, and we define
\begin{align*}
  &\lvl(\bot) := \lvl(\top) := 0 \qquad&\\
  &\lvl(A\wedge B) := \lvl(A\vee B) := \max(\lvl(A),\lvl(B))&\\
  &\lvl(A\imp B) := \max(\lvl(A) + 1,\lvl(B))&\\
  &\lvl(\forall_\alpha x.C) := \lvl(\exists_\alpha x.C) := \sup\{\lvl(C[t/x])\st t\in\ground(\alpha)\}&\\
  &\lvl(\nabla_\alpha x.C) := \lvl(C[n/x])\qquad& n\not\in\supp(C)
\end{align*}
We say that a set of definitions $\mc D$ is \emph{weakly stratified} if for every clause $H\defn_{\mc X} B\in \mc D$, and for every $\mc X$-grounding substitution $\rho$ (ie. $\rho: \emptyset\to\mc X$),

\begin{equation}
  \label{wstrat-cond}
  \lvl(H\rho)\geq \lvl(B\rho).
\end{equation}

If in addition to (\ref{wstrat-cond}), we require that for any atomic predicate $A = p\ \vec x$, and any $\vec x$-grounding substitutions $\rho_1,\rho_2$,
\begin{equation}
  \label{strat-cond}
  \lvl(A\rho_1) = \lvl(A\rho_2)
\end{equation}
which is to say that the level of a ground atomic predicate $p\ \vec t$ does not depend on its arguments, then we say that $\mc D$ is \emph{(strictly) stratified}. The (strict) stratification condition prevents $p$ from appearing negatively (ie. to the left of an implication) in $B$. For example, the definition $p \defn p\imp \bot$ is not stratified. The weak stratification condition weakens this condition by allowing $A = p\ \vec u$ to appear negatively in $B$ provided that for any ground substitution $\rho$, $\lvl(H\rho) > \lvl(A\rho)$. Examples will be provided in section \ref{sec:examples}.

Given a clause $H\defn_{\mc X} B \in\mc D$, formulas $A, B'$, and a substitution $\theta : \mc Z\to\mc Y$, we say that $\dfn(H \defn_{\mc X} B, A,\theta,B')$ if there exists a substitution $\rho:\mc Z\to \mc X$ such that $H\rho = A\theta$ and $B' = B\rho$. In addition, we require that $\supp(\theta) = \emptyset$, although the support of $\rho$ may be nonempty. Note that $B'$ is always an instance of $B$, but $A$ may not be an instance of $H$. For example, for the clause $H \defn B$ given by
$${\append\ (X :: L)\ K\ (X :: M)\defn \append\ L\ K\ M}$$
and for the formula $A = \append\ V\ \nil\ W$, the following holds
$$\dfn(H\defn B,\append\ V\ \nil\ W,[(X :: L)/V, (X :: M)/W],\append\ L\ \nil\ M)$$ for the substitution $\rho = [\nil/K].$

The introduction rules for definitions are given in Figure~\ref{fig:defn-rules}.
\begin{figure}[htb!]
\begin{mathpar}  
  {\inferrule*[right=\DL]
    {\{\mc Z;\Gamma\theta,B'\vdash C\theta\st \dfn(H\defn_\mc X B,A,\theta,B')
      \text{ and }\mc Z = \range(\theta)\}}
    {\mc Y;\Gamma, A\vdash C}}\and
  {\inferrule*[right=\DR]
    {\mc Y;\Gamma\vdash B'}
    {\mc Y;\Gamma\vdash A}}\quad \raisebox{1em}{$\dfn(H\defn_{\mc X} B,A,\epsilon_{\mc Y},B')$}
\end{mathpar}
\caption{Definition rules for a predicate $p$, provided $A = p\ \vec t$}
\label{fig:defn-rules}
\end{figure}
Conceptually, \DR corresponds to unfolding a definition for a particular instance $A$ of the head $H$, whereas \DL corresponds to case analysis on all possible instances of $A$ matching with the head $H$. 

\begin{figure}[h!]
  \begin{mathpar}
  {\inferrule*[right=\cL]
    {\inferrule*[Right=\DL,leftskip=1.5em,rightskip=1.5em]
      {\inferrule*[Right=\impL,leftskip=1.5em,rightskip=1.5em]
        {\inferrule*[Right=\axiom]{ }{p\vdash p}\\
          \inferrule*[Right=\wL]{\inferrule*[Right=\botL]{ }{\bot\vdash\bot}}{p,\bot\vdash\bot}}
        {p,p\imp\bot\vdash\bot}}
      {p,p\vdash \bot}}
    {p\vdash \bot}}\and
  {\inferrule*[right=\multc]
    { {\raisebox{-1em}{$\inferrule*[right=\DR]
        {\inferrule*[Right=\impR]
        {{\deriv{\vdots}{p\vdash\bot}}}{\vdash p\imp\bot}}
      {\vdash p}$}}\\
      {{\deriv{\vdots}{p \vdash \bot}}}}
    {\vdash \bot}
  }
\end{mathpar}
\caption{Inconsistency of the definition $p\defn p\imp \bot$}
\label{fig:inconsistency}
\end{figure}

 If the stratification condition (\ref{wstrat-cond}) is not enforced, inconsistencies may arise. Indeed, the definition $p \defn p\imp \bot$ cannot be (weakly nor strictly) stratified, and if it were allowed as a definition, we would obtain a derivation of $\vdash \bot$ as seen in Figure~\ref{fig:inconsistency}. The principle of explosion, captured by the \botL rule, then allows us to obtain a proof of any formula, making the logic inconsistent. We will see in section~\ref{sec:examples} examples of permissible definitions which are weakly, but not strictly stratified.
 
 \subsection{The fixed-point and clausal form of definitions}
 \label{sec:definition-forms}
In the previous section, we have presented a definition as a set of clauses of the form $p\ \vec t\defstar_{\mc X} B$, which we call a \emph{clausal} presentation (where $\defstar$ is either $\defn$ or $\ind$, which we will define in the next section). We now discuss an alternative way of presenting definitions. A \emph{fixed-point} presentation of a definition is a single clause of the form $p\ \vec x\defstar_{\vec x} B'\ p\ \vec x$ where each variable in $\vec x$ is unique, and $B'$ is a closed term not containing $p$. Notice that a definition in fixed-point form $p\ \vec x\defstar_{\vec x} B'\ p\ \vec x$ is already in clausal form. Conversely, a definition in clausal form
$$p\ \vec t_1\defstar_{\mc X_1} B_1\quad\ldots \quad p\ \vec t_n\defstar_{\mc X_n} B_n$$
may be encoded into a fixed-point form as follows
\begin{align*}
  \eq\ x\ x &\defn_{\{x\}} \top\\
  p\ \vec x &\defstar_{\vec x} (\exists \mc X_1.(\eq\ x_1\ t^1_1)\wedge\ldots\wedge(\eq\ x_k\ t^k_1) \wedge B_1)\vee\ldots\vee\\
            &\qquad(\exists \mc X_n.(\eq\ x_1\ t^1_n)\wedge\ldots\wedge(\eq\ x_k\ t^k_n)\wedge B_n)
\end{align*}
where $\vec t_i = t_i^1,\ldots,t_i^k$ for each $i$, with $t_i^j$ lying over $\mc X_i$. For example, $\listp$ and $\append$ from Figure~\ref{fig:append} may be specified in fixed-point form in the following way
\begin{equation}
  \label{def:listp}
\begin{aligned}
  \listp\ L &\ind (\eq\ L\ \nil)\vee(\exists X,L'. (\eq\ L\ (X :: L'))\wedge (\listp\ L'))\\
  \append\ L\ K\ M &\defn (\eq\ L\ \nil \wedge \eq\ K\ M)\vee\\
                   &\qquad(\exists X,L',M'. (\eq\ L\ (X :: L'))\wedge (\eq\ M\ (X :: M'))\wedge (\append\ L'\ K\ M'))
\end{aligned}
\end{equation}

However, while these two presentations appear interchangeable, a weakly stratified definition in clausal form may no longer be weakly stratified when converted to its fixed-point form, which is why we require that definitions be presented in clausal form. For example, if we stratify $\listp$ with the level $\lvl(\listp\ L) := |L|$ where for closed terms $X$ and $L$, 
\begin{mathpar}
  |\nil| := 0\and |X :: L| := 1 + |L|
\end{mathpar}
then while its clausal form in Figure~\ref{fig:append} is weakly stratified, its fixed-point form in Figure~\ref{def:listp} is not because for any ground term $L$, the level of the ground formula
$$\exists X,L'.(\eq\ L\ (X :: L'))\wedge(\listp\ L')$$
is the infinite ordinal $\omega$, which is greater than $\lvl(\listp\ L) = |L|$. However, this does not preclude the existence of another assignment of levels for which both presentations of $\listp$ are stratified.

In the next section, we will introduce inductive definitions using the fixed-point presentation. Since inductive definitions will be required to be (strictly) stratified, we will present them in fixed-point form, which will simplify the rules for inductive definitions.

\subsection{Inductive definitions}

\begin{figure}[h]
\begin{mathpar}
  {\inferrule*[right=\IL]
    {\vec x; B\ S\ \vec x\vdash S\ \vec x\qquad \mc X; \Gamma,S\ \vec t\vdash C}
    {\mc X; \Gamma, p\ \vec t\vdash C}}\and
  {\inferrule*[right=\IR]
    {\mc X;\Gamma\vdash B\ p\ \vec t}
    {\mc X;\Gamma\vdash p\ \vec t}}
\end{mathpar}
\caption{Inductive definition rules for $p\ \vec x\ind_{\vec x} B\ p\ \vec x\in \mc I$ and inductive invariant $S$}
\label{fig:ind-rules}
\end{figure}

A valid set of inductive definitions $\mc I$ consists of clauses of the form $p\ \vec x\ind_{\vec x} B\ p\ \vec x\in \mc I$ for some predicate constant $p : \omega$ in the signature $\Sigma$ and a closed term $B$ of type $\omega\to\omega$ not containing $p$. We assume all variables in $\vec x$ are distinct, which allows us to uniquely determine $B$ from the formula $B\ p\ \vec x$. We assume $p$ is in the head of at most one clause in $\mc I$ and is not already defined in $\mc D$. The \emph{inductive definition} associated with $p$ is the unique clause in $\mc I$ with $p$ in the head. An \emph{inductive invariant} for an inductive definition is a closed term $S$ of type $\omega$. We require that the set of inductive definitions be (strictly) stratified, and present the introduction rules for inductive definitions in Figure~\ref{fig:ind-rules}. For example, the \muL rule for the $\listp$ predicate defined in \ref{def:listp} is given by the following
\begin{mathpar}
  {\inferrule*[Right=\IL]
    {{L;\eq\ L\ \nil\vee(\exists X,L'.(\eq\ L\ (X :: L')) \wedge (S\ L'))\vdash S\ L}\\
    \\\mc X;\Gamma,S\ \ell\vdash C}
  {\mc X;\Gamma,\listp\ \ell\vdash C}}
\end{mathpar}
for any inductive invariant $S$.


In section~\ref{examples:ind-counterexample} we will see an example showing that the weak stratification condition on inductive definitions is not sufficient to ensure consistency.

\subsection{The logic \LDIN}
Now that we have defined every component to our logic, we may finally define
\begin{definition}[The logic \LDIN]
  Given a valid set of definitions $\mc D$ and a valid set of inductive definitions $\mc I$ such that
  \begin{enumerate}
  \item $\mc D$ is weakly stratified, and
  \item $\mc I$ is (strictly) stratified
  \end{enumerate}
  derivability in the logic $\text{\LDIN}\!(\mc D,\mc I)$ is described inductively by the combination of the rules in Figures~\ref{fig:core-rules},\ref{fig:struct-rules},\ref{fig:defn-rules}, and \ref{fig:ind-rules}. When $\mc D$ and $\mc I$ are clear from the context we simply write \LDIN.
\end{definition}

Given an inference rule, if it is a right introduction rule, all of its premises are called \emph{major} premises. If it is a left introduction rule or a multicut, then only those with the same consequent as its conclusion are major premises. All other premises are \emph{minor} premises. In particular, we note that the only rules with minor premises are \impL, \muL, and \multc.

}

\section{Examples}
\label{sec:examples}
{\newcommand{\z}{z}
\newcommand{\suc}{s}
\newcommand{\arr}{arr}
\newcommand{\app}{app}
\newcommand{\lam}{lam}
\newcommand{\nt}{nt}

\newcommand{\ev}{ev}
\newcommand{\odd}{odd}

\newcommand{\logeq}{logeq}
\newcommand{\steps}{steps}
\newcommand{\step}{step}
\newcommand{\natp}{nat}
\newcommand{\red}{red}
\newcommand{\sn}{sn}

\newcommand{\eq}{eq}
\newcommand{\true}{\mathrm{T}}
\newcommand{\false}{\mathrm{F}}
\newcommand{\lbox}{\square}
\renewcommand{\diamond}{\lozenge}

\newcommand{\tm}{\mathrm{tm}}
\newcommand{\ty}{\mathrm{ty}}

We have already seen that inconsistencies may appear in the logic of definitions as soon as the predicate being defined appears negatively in the body of its definition. A sufficient condition to prevent this is the stratification condition (\ref{strat-cond}), but the aim of this work is weaken this to the weak stratification condition (\ref{wstrat-cond}) to allow for a broader body of definitions. In this section, we will see examples of definitions which are weakly but not strictly stratified, thus demonstrating the usefulness of such definitions. Our first example (from \cite{Tiu12}) gives a simple but illustrative example for defining even numbers. Our second example motivates weak stratification as a means of encoding logical equivalence, a special case of a more general method to prove properties about programs known as logical relations \cite{Ha16}. Finally, we give a counter-example demonstrating that weak stratification is not a sufficient condition for inductive definitions to be consistent, thus justifying our restriction on inductive definitions.

In addition to the examples given in this section, we note that weakly stratified definitions may be used to encode the G\"odel-Gentzen translation, as well as Kolmogorov's double negation translation of classical logic in intuitionistic logic, as described in \cite{Tiu12}. As pointed out in \cite{BaNa12}, another useful application of such definitions is for the formalization of Tait-style \cite{Ta67, GTL89} strong normalization arguments. A more elaborate example application of a logical relations argument may be found in \cite{WN16}, which implicitly uses weak stratification to define a step-indexed logical equivalence relation to prove correctness of compiler transformations for functional programs.

\subsection{Even numbers}
\label{examples:even}
We consider the following definition recognizing all even natural numbers by first specifying the following signature over the base type $\tm$:
\begin{mathpar}
  \z : \tm\and \suc : \tm\to\tm\and \ev : \tm\to o
\end{mathpar}
We then specify the definition clauses as follows
\begin{mathpar}
  \ev\ \z \defn \top\and
  \ev\ (\suc\ X) \defn (\ev\ X)\imp \bot
\end{mathpar}
This definition cannot be stratified because $\ev\ X$ appears negatively in the body of its own definition. However, it can be weakly stratified by assigning the levels $\lvl(\ev\ t) := |t|$ for any ground $t$ where $|\z| := 0$ and $|\suc\ X| := |X| + 1$. One might want to compare this to a more standard way of inductively defining even numbers, ie.
\begin{mathpar}
  \ev'\ X \ind (X = \z) \vee (\exists Y. \eq\ X\ (\suc\ (\suc\ Y)) \wedge (\ev'\ Y))
\end{mathpar}
It is then provable in \LDIN that $\forall x. \ev'\ x \imp \ev\ x$. Note however that since $\ev$ cannot be stratified, its definition cannot be made inductive, and thus the converse need not be provable.

\subsection{Logical equivalence}
\label{examples:logeq}
For a more interesting example, we show how we can encode the logical equivalence relation of programs presented in Figure~\ref{fig:logeq} (see eg. Harper \cite{Ha16}) using weak stratification. In particular, we encode types and terms in the simply typed lambda calculus described by the constants:
\begin{mathpar}
  \z : \tm\and \suc : \tm\to\tm\and
  \app : \tm\to\tm\to\tm\and \lam : (\tm\to\tm)\to\tm\\
  \nt : \ty\and \arr : \ty\to\ty\to\ty
\end{mathpar}

We then define a predicate $\natp : \tm\to o$ specifying natural number values, as well as a predicate $\step : \tm\to\tm\to o$ specifying a single computation step:
\begin{mathpar}
  \natp\ \z \defn \top\and
  \natp\ (\suc\ X) \defn \natp\ X\and\\
  \step\ (\suc\ X)\ (\suc\ X')\defn \step\ X\ X'\and
  \step\ (\app\ (\lam\ S)\ U)\ (S\ U) \defn \top\and
  \step\ (\app\ S\ T)\ (\app\ S'\ T)\defn \step\ S\ S'\and
  \step\ (\app\ S\ T)\ (\app\ S\ T')\defn \step\ T\ T'
\end{mathpar}

and we define the relation $\steps : \tm\to\tm\to o$ to be the reflexive-transitive closure of $\step$:
\begin{mathpar}
  \steps\ T\ T \defn \top\and
  \steps\ S\ T \defn \exists y.\step\ S\ y \wedge \steps\ y\ T
\end{mathpar}

Finally, logical equivalence is encoded by the relation $\logeq : \ty\to\tm\to\tm\to o$ as follows:
\begin{align*}
  \logeq\ \nt\ S\ T &\defn (\exists n. (\steps\ S\ n) \wedge (\steps\ T\ n)\wedge (\natp\ n))\\
  \logeq\ (\arr\ A\ B)\ S\ T &\defn (\forall u_1,u_2. (\logeq\ A\ u_1\ u_2)\imp (\logeq\ B\ (\app\ S\ u_1)\ (\app\ T\ u_2)))
\end{align*}

This example cannot be stratified because $\logeq$ appears negatively in its body. However, we can assign the following levels to each of the defined predicates for arbitrary ground terms $A, B, S, T,$ and $N$ as follows:
\begin{mathpar}
  \lvl(\steps\ S\ T) := 0\and
  \lvl(\step\ S\ T) := 0\\
  \lvl(\natp\ N) := 0\and
  \lvl(\logeq\ A\ S\ T) := |A|
\end{mathpar}
where we define $|\nt| := 0$ and $|\arr\ A\ B| := \max(|A| + 1,|B|)$. We may check that these definitions are weakly stratified with respect to this level assignment. The only nontrivial condition to check is that for any ground $A, B, S$ and $T$,
\begin{align*}
  \lvl(\logeq\ (\arr\ A\ B)\ S\ T)
  &= \max(|A|+1,|B|) \\
  &=\sup_{u_1,u_2}\{\max(\lvl(\logeq\ A\ u_1\ u_2) + 1, \lvl(\logeq\ B\ (\app\ S\ u_1)\ (\app\ T\ u_2)))\}\\
  &=\sup_{u_1,u_2}\{\lvl((\logeq\ A\ u_1\ u_2)\imp(\logeq\ B\ (\app\ S\ u_1)\ (\app\ T\ u_2)))\}\\
  &\geq \lvl(\forall u_1,u_2. (\logeq\ A\ u_1\ u_2)\imp (\logeq\ B\ (\app\ S\ u_1)\ (\app\ T\ u_2)))
\end{align*}
thus $\logeq$ is weakly stratified, and therefore a valid definition.



\subsection{Weak stratification is insufficient for inductive definitions}
\label{examples:ind-counterexample}

Recall that we required inductive definitions to be (strictly) stratified. We now give an example showing that if the strict stratification condition on inductive definitions is not enforced, the logic may become inconsistent. Indeed, consider the signature with base type $\tm$ and definition $\ev : \tm\to o$ in the example in section~\ref{examples:even}. We additionally suppose we are allowed to define the predicate $\odd : \tm\to o$ inductively as follows
\begin{align*}
  \odd\ (\suc\ X) \ind (\odd\ X)\imp\bot
\end{align*}
which corresponds to the fixed-point presentation $\odd\ X \ind B\ \odd\ X$ with operator
$$B \equiv \lambda S. \lambda x. \exists y. (\eq\ x\ (\suc\ y)\wedge(S\ y\imp\bot)).$$

We note firstly that, using induction on $\odd$ with the inductive invariant $S = \lambda x. \ev\ x$, we derive the sequent $X; \odd\ X \vdash \ev\ X$ as follows:
\begin{equation}
  \label{counter-ex1}
  \inferrule*[right=\muL]
  {{\inferrule*[right=\exL;\andL{*}]
      {\inferrule*[Right=\DL]
        {\inferrule*[Right=\DR]
          {\inferrule*[Right=\init]{ }{Y; \ev\ Y\imp\bot\vdash\ev\ Y\imp\bot}}
          {Y; \ev\ Y\imp\bot\vdash\ev\ (\suc\ Y)}}
        {X, Y; \eq\ X\ (\suc\ Y), \ev\ Y\imp\bot\vdash\ev\ X}}
    {X; \exists y. \eq\ X\ (\suc\ y) \wedge \ev\ y\imp\bot\vdash\ev\ X}}\\
    \inferrule*[Right=\axiom]
    { }{X; \ev\ X\vdash \ev\ X}}
  {X; \odd\ X\vdash \ev\ X}
\end{equation}

Second, we use induction on $\odd$ with the inductive invariant $S = \lambda x.(\ev\ x\imp\bot)$ to derive the sequent $X; \odd\ X\vdash \ev\ X\imp\bot$ as follows:
\begin{equation}
  \label{counter-ex2}
  \inferrule*[right=\muL]
  {{\inferrule*[right=\exL;\andL{*}]
      {\inferrule*[Right=\DL]
        {\inferrule*[Right=\impR]
          {\inferrule*[Right=\DL]
            {\inferrule*[Right=\impL, leftskip=5em, rightskip=5em]
              {{\inferrule*[right=\init]{ }{Y; \ev\ Y\imp\bot\vdash \ev\ Y\imp\bot}}\\
                {\inferrule*[Right=\axiom;\wL]{ }{Y; \bot, \ev\ Y\imp\bot\vdash\bot}}}
              {Y; (\ev\ Y\imp\bot)\imp\bot, \ev\ Y\imp\bot\vdash\bot}}
            {Y; (\ev\ Y\imp\bot)\imp\bot, \ev\ (\suc\ Y)\vdash\bot}}
          {Y; (\ev\ Y\imp\bot)\imp\bot\vdash \ev\ (\suc\ Y)\imp\bot}}
        {X, Y; \eq\ X\ (\suc\ Y), (\ev\ Y\imp\bot)\imp\bot\vdash \ev\ X\imp\bot}}
      {X; \exists y. \eq\ X\ (\suc\ y) \wedge (\ev\ Y\imp\bot)\imp\bot\vdash \ev\ X\imp\bot}}
    \inferrule*[Right=\init]
    { }{X; \ev\ X\imp\bot\vdash \ev\ X\imp\bot}}
  {X; \odd\ X\vdash \ev\ X\imp\bot}
\end{equation}

With a single use of the multicut rule applied to the sequents derived in \ref{counter-ex1} and \ref{counter-ex2} above, we obtain
\begin{equation}
  \label{counter-ex3}
  X; \odd\ X\vdash \bot
\end{equation}
Specializing $X$ to $\z$ in \ref{counter-ex3}, we obtain a derivation of $\odd\ \z \vdash \bot$, and thus a derivation of $\vdash \odd\ (\suc\ \z)$ through the \DR and \impR rules. Finally, specializing $X$ to $\suc\ \z$ in \ref{counter-ex3}, we obtain $\odd\ (\suc\ \z)\vdash\bot$, which when cut with the previous result produces a derivation of $\vdash\bot$.

}

\section{The ground logic \LDINinf}
\label{sec:ldininf}
{
We now define a ground version of \LDIN which we call \LDINinf. We will then define an interpretation of \LDIN in \LDINinf, and show that the cut rule is admissible in \LDINinf. This will imply that any ground derivation of $\vdash \bot$ has a cut-free derivation, which we will see is impossible. Because the interpretation preserves derivability, this will allow us to conclude that $\vdash \bot$ is not derivable in \LDIN, and therefore that \LDIN is \emph{consistent}.

The formulas in \LDINinf are the formulas in \LDIN which are well formed over the empty variable context. Since all formulas are ground, we restrict the sequent $\mc X;\Gamma\vdash C$ to have no eigenvariables, which translates to requiring that $\mc X = \emptyset$. We therefore denote a ground sequent by $\Gamma\vdash C$. The rules of \LDINinf are the same as those of \LDIN except for \allR, \exL, and \muL. We replace the \allR and \exL rules with the following infinitary rules in Figure~\ref{fig:omega-rules}.

\begin{figure}[h]
  \begin{mathpar}
    \inferrule*[right=\allR]
    {\left\{{\Gamma\vdash C[t/x]}
      \right\}_{t\in\ground(\alpha)}}
    {\Gamma\vdash \forall_\alpha x.C}\and
    \inferrule*[right=\exL]
    {\left\{{\Gamma, B[t/x]\vdash C}
      \right\}_{t\in\ground(\alpha)}}
    {\Gamma, \exists_\alpha x.B \vdash C}
  \end{mathpar}
  \caption{\allR and \exL rules in \LDINinf}
  \label{fig:omega-rules}
\end{figure}

The definition \DL and \DR rules in \LDINinf are the same as those of \LDIN. The restriction to ground conclusions allows us to write the rules as in Figure~\ref{fig:ground-def-rules}.
\begin{figure}
  \begin{mathpar}
  {\inferrule*[right=\DL]
    {\{\Gamma,B'\vdash C\st \dfn(H\defn B,A,\epsilon_\emptyset,B')\}}
    {\Gamma, A\vdash C}}\and
  {\inferrule*[right=\DR]
    {\Gamma\vdash B'}
    {\Gamma\vdash A}}\quad\raisebox{1em}{$\dfn(H\defn_{\mc X} B, A, \epsilon_\emptyset,B')$}
\end{mathpar}
\caption{Definition rules in \LDINinf for a predicate $p$, provided $A = p\ \vec t$}
\label{fig:ground-def-rules}
\end{figure}
However, we must check that these rules are well-formed, which amounts to checking that the premises are well-formed over the empty variable context. Equivalently, we must check that all the formulas in each premise is ground. Suppose we are given a ground context $\Gamma$, a ground formula $A$, and clause $H\defn_{\mc X} B$ such that $\dfn(H\defn_{\mc X}B, A,\epsilon_\emptyset,B')$, then there exists $\rho : \mc Z\to X$ such that $A = A\epsilon_\emptyset = H\rho$ and $B\rho = B'$. Since $H$ lies in the $L_\lambda$ pattern fragment and $A$ is ground, this implies that $\rho$ is unique and furthermore that it must be ground. Since $B$ lies over $\mc X$, which is the domain of $\rho$, it follows that $B'$ must be ground, as desired.

\begin{figure}[h!]
  \begin{mathpar}
    \inferrule*[right=\IL]
    {\left\{{B\ S\ \vec t\vdash S\ \vec t}
      \right\}_{\vec t\in\ground(\vec \alpha)}\\
      {\Gamma, S\ \vec u\vdash C}}
    {\Gamma, p\ \vec u\vdash C}
\end{mathpar}
\caption{\IL rule in \LDINinf for $p\ \vec x\ind_{\vec x} B\ p\ \vec x\in \mc I$ and inductive invariant $S$}
\label{fig:ind-inf-rule}
\end{figure}
Finally, for any inductive definition $p\ \vec x\ind_{\vec x} B\ p\ \vec x\in\mc I$ and inductive invariant $S$, we introduce a new ground version of the \IL rule given in Figure~\ref{fig:ind-inf-rule}.

\begin{definition}[The logic \LDINinf]
  Given a valid set of definitions $\mc D$ and a valid set of inductive definitions $\mc I$ such that
  \begin{enumerate}
  \item $\mc D$ is weakly stratified, and
  \item $\mc I$ is (strictly) stratified
  \end{enumerate}
  Derivability in the logic $\text{\LDINinf}\!(\mc D,\mc I)$ is described inductively by the combination of the rules in Figures~\ref{fig:core-rules}, \ref{fig:struct-rules}, \ref{fig:defn-rules}, and \ref{fig:ind-rules} where \allR and \exL are replaced by those in Figures~\ref{fig:omega-rules}, and the \IL rule is replaced by that in Figure~\ref{fig:ind-inf-rule}. When $\mc D$ and $\mc I$ are clear from the context we simply write \LDINinf.
\end{definition}

}

\subsection{Interpreting \LDIN\ in \LDINinf}
\label{sec:interp}
{
The goal is now to define an interpretation of \LDIN in \LDINinf, which will relate the derivability of a sequent in \LDIN to the derivability of a ground sequent in \LDINinf. This will allow us to reduce the consistency of the former to that of the latter. We first define the \emph{height} of a derivation $\Pi$ with premise derivations $\{\Pi_i\}_{i\in\mc I}$ by
$$\height(\Pi) = \sup_{i\in\mc I}\{\height(\Pi_i) + 1\}.$$
In particular if $\Pi$ has no premise derivations, $\height(\Pi) = 0$, and furthermore we have $\height(\Pi_i) < \height(\Pi)$ for each $i\in\mc I$. Note that $\mc I$ may be infinite, so the height is an ordinal.

The interpretation is specified by the following lemma.
\begin{restatable}[Grounding Lemma]{lemma}{lemgroundinglemma}
  If $\mc Y;\Gamma\vdash C$ is derivable in \LDIN, then for any $\mc Y$-grounding substitution $\delta : \emptyset\to\mc Y$, the ground sequent $\Gamma\delta\vdash C\delta$ is derivable in \LDINinf. 
\end{restatable}
\begin{proof}
  The proof is by induction on the height of the derivation of $\mc Y;\Gamma\vdash C$. Note that we do not require $\supp(\delta)$ to be empty. We only show the cases for \allR and \DL below.
  \case{\allR} Suppose the derivation ends with \allR with $C = \forall_\alpha x.D$.
  Given $y\not\in\mc Y$ and $\vec n = \supp(D)$, suppose $\delta$ is an $\mc Y$-grounding substitution. Since $x$ may be renamed in $\forall_\alpha x.D$, we will assume that $x\not\in \mc Y$. Now for every $t\in\ground(\alpha)$, notice that
  $$(D[y\ \vec n/x])[\delta,(\lambda\vec x.t)/y] = (D[\delta,y\ \vec n/x])[\lambda\vec x.t/y] = (D\delta)[t/x].$$

  Since $\delta' = [\delta,(\lambda\vec x.t)/y]$ is an $\mc Y,y$-grounding substitution, we may apply it to the premise $\mc Y,y;\Gamma\vdash D[y\ \vec n/x]$ to obtain a ground derivation of $\Gamma\delta'\vdash (D\delta)[t/x]$ by the induction hypothesis, as follows
    \begin{mathpar}
    {\inferrule*[right=\allR]
      {\mc Y,y;\Gamma\vdash D[y\ \vec n/x]}{\mc Y;\Gamma\vdash \forall_\alpha x.D}}
    \quad\leadsto\quad
    {\inferrule*[right=\allR]
      {\left\{\Gamma\delta'\vdash (D\delta)[t/x])\right\}_{t\in\ground(\alpha)}}
      {\Gamma\delta\vdash \forall_\alpha x.(D\delta)}}
    \end{mathpar}

  \case{\DL} Suppose the derivation ends with \DL on the atom $A$ where $\Gamma = \Gamma', A$. It is clear that if $\dfn(H\defn_{\mc X} B, A\delta, \epsilon_\emptyset, B')$, then $\dfn(H\defn_{\mc X} B, A, \delta, B')$, and furthermore that $\range(\delta) = \emptyset$ since $\delta$ is a ground substitution. Thus an \LDINinf derivation of a premise $\Gamma'\delta, B'\vdash C\delta$ below is obtained by applying the induction hypothesis to the ground premise $\emptyset;\Gamma'\delta,B'\vdash C\delta$ with the empty grounding substitution. 

    \begin{mathpar}
    {\inferrule*[right=\DL]
      {\{\mc Z; \Gamma'\theta, B'\vdash C\theta\st \dfn(H\defn_{\mc X} B, A, \theta, B')
        \text{ and }\mc Z = \range(\theta)\}}
      {\mc Y; \Gamma', A \vdash C}
    }\quad\leadsto\quad
    {\inferrule*[right=\DL]
      {\{\Gamma'\delta, B' \vdash C\delta\st\dfn(H\defn_{\mc X} B, A\delta, \epsilon_\emptyset,B')\}}
      {\Gamma'\delta, A\delta \vdash C\delta}
    }
    \end{mathpar}

Finally, we remark that although the grounding substitution $\delta$ may have nonempty support, this is not a limitation for the \nabL and \nabR cases since we can always rename the nominals in a sequent while preserving derivability.
\end{proof}
A detailed proof of this lemma can be found in Appendix~\ref{appendix:grounding-lemma}.

\subsection{The consistency of \LDINinf implies consistency of \LDIN}

As a result of the grounding lemma, we obtain the following
\begin{lemma}[Interpretation lemma]
  \label{conldin}
  If \LDINinf is consistent, then so is \LDIN.
\end{lemma}
\begin{proof}
  From grounding lemma, if $\vdash \bot$ is derivable in \LDIN, then it is also derivable in \LDINinf. Thus, the contrapositive also holds: If $\vdash \bot$ is not derivable in \LDINinf, then it is not derivable in \LDIN.
\end{proof}
}

\section{Cut reductions for \LDINinf}
\label{sec:ldininf-cutred}
{

We now define the cut reduction relation on derivations in \LDINinf. This specifies how the cut rule may be propagated upwards in a derivation tree. In section~\ref{sec:cutelim} we will see how this relation can be used to obtain cut-free derivations.

Specifically, we define a {\em reduction} relation between derivations, following closely
the reduction relation in \cite{MDM00}. The \emph{redex}, that is the derivation to be reduced, is always a derivation $\Xi$ ending with the multicut rule
\begin{mathpar}
{\inferrule*[right=\multc]
        {\deriv{\Pi_1}{\Delta_1\vdash B_1}
        \\ \cdots
        \\ \deriv{\Pi_n}{\Delta_n \vdash B_n}
        \\ \deriv{\Pi}{B_1,\ldots,B_n,\Gamma\vdash C}
        }
        {\Delta_1,\ldots,\Delta_n,\Gamma\vdash C}}
\end{mathpar}

We refer to the formulas $B_1,\dots,B_n$ in the multicut as {\em cut formulas}. If a left or right rule introduces a cut formula, we say the rule is \emph{principal}. If $\Pi$ reduces to $\Pi'$, we say that $\Pi'$ is a \emph{reduct} of $\Pi$, and we may write $\Pi\to\Pi'$.

Specifically, the reduction relation relates a derivation $\Xi$ ending with the multicut rule \multc, to a new derivation $\Xi'$, which may again end with a multicut which we will see later is always of smaller complexity. The reduction relation is specified by case analysis on the last rule of $\Pi$ and $\Pi_i$ for a chosen cut formula $B_i$. The main idea is to group the cases into mutually exclusive {\em essential cases}, {\em inductive cases}, {\em left/right commutative cases}, {\em structural cases}, {\em left/right multicut cases}, and {\em left/right axiom cases}. We summarize these cases below. The full reduction relation is specified in detail Appendix~\ref{appendix:cut-red}.

\case{\em \underline{Essential cases}}
An essential case occurs when $\Pi_i$ ends with a right introduction rule and $\Pi$ ends with a left introduction rule. For example, here is the case when $\Xi$ is a cut for which $\Pi_1$ ends with \andR, and $\Pi$ ends with \andL{i}:

If $\Pi_1$ and $\Pi$ are
 \begin{mathpar}
    {\inferrule*[right=\andR]
      {\deduce{\Delta_1\vdash B_1^1}
        {\Pi_1^1}
        \\ \deduce{\Delta_1\vdash B_1^2}
           {\Pi_1^2}}
      {\Delta_1\vdash B_1^1 \wedge B_1^2}}
    \and
    {\inferrule*[right=\andL{i}]
      {\deduce{B_1^i,B_2,\ldots,B_n,\Gamma\vdash C}
        {\Pi'}}{B_1^1 \wedge B_1^2,B_2,\ldots,B_n,\Gamma\vdash C}}
  \end{mathpar}
  then $\Xi$ reduces to the derivation
  \begin{mathpar}
    {\inferrule*[right=\multc]
      {{\deriv{\Pi_1^i}{\Delta_1\vdash B_1^i}}\\ \ldots \\
       {\deriv{\Pi'}{B_1^i,B_2,\ldots,B_n,\Gamma\vdash C}}}
      {\Delta_1,\ldots,\Delta_n,\Gamma\vdash C}
    }
  \end{mathpar}
  which we denote $\multcut(\Pi_1^i,\Pi_2,\dots,\Pi_n,\Pi')$.
  Notice here (as with the other essential cases) that the cut now occurs on the formula $B_1^i$, which has complexity no greater than $B_1^1\wedge B_1^2$. In this case, the decreasing complexity is captured by the level of the cut formula.

  \case{\em \underline{Inductive cases}}
 An inductive case occurs when $\Pi$ ends with \muL introducing a cut formula $p\ \vec t$ for an inductive definition $p\ \vec x \ind_{\vec x} B\ p\ \vec x$. Suppose the cut formula being introduced is $B_1 = p\ \vec t$. In this case, $\Pi_1$ is a derivation of the sequent $\Delta_1\vdash p\ \vec t$ and $\Pi$ is the following derivation for some inductive invariant $S$
\begin{mathpar}
  {\inferrule*[right=\IL]
    {\left\{\deriv{\Pi^{\vec u}_S}{B\ S\ \vec u\vdash S\ \vec u}\right\}_{\vec u\in\ground(\vec \alpha)}\\
      \deriv{\Pi'}{S\ \vec t,B_2,\ldots,B_n,\Gamma\vdash C}}
    {\Seq{p\ \vec t,B_2,\ldots,B_n,\Gamma}{C}}}
\end{mathpar}
The key idea is that since we know that $S$ is a pre-fixed-point of $p$, then any time $p\ \vec t$ holds we should be able to obtain $S\ \vec t$. More specifically, the family of derivations $\{\Pi^{\vec u}_S\}_{\vec u}$, which we abbreviate by $\Pi_S$, may be used to obtain a derivation of $\Delta_1\vdash S\ \vec t$. We call this derivation the \emph{unfolding} of the derivation $\Pi_1$ with respect to $\Pi_S$, which we denote by $\mu(\Pi_1,\Pi_S)$. We note that unfolding a derivation crucially depends on the inductive definition $p$ to be strictly stratified.

We can now reduce $\Xi$ to the following

\begin{mathpar}
  {\inferrule*[right=\multc]
    {{\deriv{\mu(\Pi_1,\Pi_S)}{\Delta_1\vdash S\ \vec t}}\\
      \ldots\\ 
      {\deriv{\Pi'}{S\ \vec t,B_2,\ldots,B_n,\Gamma\vdash C}}}
    {\Delta_1,\ldots,\Delta_n, \Gamma\vdash C}}
\end{mathpar}

  \case{\em \underline{Left commutative cases}}
  A left commutative case occurs when $\Pi_i$ ends with a rule acting on a formula other than a cut formula. The idea is to push the cut above the last rule of $\Pi_i$. Consider for example the case in which $\Pi_i$ ends with \orL and $\Pi$ ends with some left introduction rule other than \cL or \wL, that is suppose $\Pi_1$ is
  \begin{mathpar}
\inferrule*[right=\orL]
{{\deriv{\Pi_1^1}{{\Delta_1', C_1}\vdash{B_1}}}\\
  {\deriv{\Pi_1^2}{{\Delta_1', C_2}\vdash{B_1}}}
}{\Delta_1', C_1\vee C_2\vdash B_1}
\end{mathpar}
then let $\Xi^i = mc(\Pi^i_1,\Pi_2,\dots,\Pi_n,\Pi)$, then $\Xi$ reduces to
\begin{mathpar}
\inferrule*[right=\orL]
{{\deriv{\Xi^1}{\Seq{\Delta_1', C_1,\Delta_2,\ldots,\Delta_n,\Gamma}{C}}}\\
  {\deriv{\Xi^2}{\Seq{\Delta_1', C_2,\Delta_2,\ldots,\Delta_n,\Gamma}{C}}}
}{\Seq{\Delta_1', C_1\vee C_2,\Delta_2,\ldots,\Delta_n,\Gamma}{C}}
\end{mathpar}

\case{\em \underline{Right commutative cases}}
  Similarly to a left commutative case, a right commutative case occurs when $\Pi$ ends with a rule acting on a formula other than a cut formula. In this case, the multicut is pushed above the last rule in $\Pi$, reducing the height of the major premise in the reduct.

\case{\em \underline{Structural cases}} This is the case in which $\Pi$ ends with a structural contraction or weakening rule (\cL or \wL) on a cut formula. In this case, the multicut is pushed above the structural rule in $\Pi$. For example, if $\Pi$ ends with a contraction on the cut formula $A_1$ with premise derivation $\Pi'$, then the multicut $\multcut(\Pi_1,\ldots,\Pi_n,\Pi)$ reduces to the multicut $\multcut(\Pi_1,\Pi_1,\ldots,\Pi_n,\Pi')$ followed by some contraction rules.

\case{\em \underline{Left multicut case}} This is the case in which $\Pi_i$ ends with a multicut. In this case, if $\Pi_i$ reduces to $\Pi_i'$, then $\Xi$ is reduced by reducing $\Pi_i$ to $\Pi_i'$. 

\case{\em \underline{Right multicut case}} This is the case in which $\Pi$ ends with a multicut, in which case the multicut in $\Pi$ is merged with the multicut in $\Xi$, reducing the height of the major premise in the reduct.

\case{\em \underline{Left axiom case}}
The left axiom case occurs when $\Pi_i$ ends with the axiom rule. In this case, the multicut $\multcut(\Pi_1,\ldots,\Pi_n,\Pi)$ reduces to $\multcut(\Pi_2,\ldots,\Pi_n,\Pi)$, reducing the number of cut formulas in the reduct.

\case{\em \underline{Right axiom case}}
The right axiom case occurs when $\Pi$ ends with the axiom rule. In this case, we must have $n = 1$, and the multicut $\multcut(\Pi_1,\Pi)$ simply reduces to $\Pi_1$, where nominals may have been renamed.

}

\section{Cut elimination for \LDINinf}
\label{sec:cutelim}
{
Our aim is now to show how the cut reduction relation can be used to eliminate cuts from a derivation. The key to doing this will be to show that the reduction relation is well-founded on \LDINinf derivations. We capture this idea with the intermediate notion of {\em normalizable} derivation. Showing that every derivation is normalizable will allow us to conclude at the end of this section that the (multi)cut rule may be eliminated from \LDINinf. The main technique used to prove normalizability is the notion of \emph{reducible} derivation, a technique inspired by Tait and Martin-L\"of (see \cite{Ta67,ML71}), and adapted to the current context by McDowell and Miller in \cite{MDM00}. Our proof closely follows that of \cite{Tiu12} and \cite{MDM00}, together with the treatment of induction from \cite{Tiu04,MT04}.

\subsection{Normalizability}
\label{sec:normalizability}
We start with the following inductive definition. Note here that by inductive definition, we mean the smallest set of derivations closed under the specified operation (see definition 1.1.1 in \cite{A77} for a precise definition).

\begin{definition}[normalizability]
  The set of normalizable derivations in \LDINinf is inductively defined as follows:
  \begin{enumerate}
  \item if $\Pi$ ends with a multicut, then $\Pi$ is normalizable if every reduct $\Pi'$ is normalizable
  \item otherwise, $\Pi$ is normalizable if each of its premise derivations is normalizable.
  \end{enumerate}
\end{definition}

The goal is to show that every derivation in \LDINinf is normalizable. 
To do this we need the intermediate notion of $\gamma$-reducibility, which is defined by transfinite recursion on the level $\gamma$ of a derivation, which is an ordinal. The \emph{level} of a derivation of a (ground) sequent $\Gamma\vdash C$ is defined to be $\lvl(C)$. Note therefore that if $\Pi$ reduces to $\Pi'$, then $\lvl(\Pi) = \lvl(\Pi')$. Furthermore, note that levels are defined so that every rule has a non-increasing consequent in each of its major premises. 
In the case of the \DR and \muR rules, this condition crucially depends on the weak and strict stratification conditions, respectively. Finally we note that below, $\gamma$-reducibility of a derivation ending with the \impR rule depends on $\alpha$-reducibility to already be defined for some $\alpha < \gamma$. These observations guarantee that the following is well defined.

\begin{definition}[reducibility]
  Define $\gamma$-reducibility of a derivation $\Pi$ inductively\footnote{In the notation of \cite{A77}, we recursively define a family of operators $\Phi_\gamma$ with the rules (1) $\{\Pi'\st \Pi\to\Pi'\}\rightarrowtriangle\Pi$, (2) $\{\Pi'\}\cup\{\multcut(\Upsilon,\Pi')\st \R_\alpha(\Upsilon)\}\rightarrowtriangle\Pi$ where $\R_\alpha(\Upsilon)$ means $\Upsilon$ is $\alpha$-reducible, and (3-5) $\{\Pi'\st\Pi'\triangleright\Pi\}\rightarrowtriangle\Pi$ where $\Pi'\triangleright\Pi$ means that $\Pi'$ is a major premise derivation of $\Pi$, and all of $\Pi$'s minor premise derivations are normalizable. The set of $\gamma$-reducible derivations is the least set closed under $\Phi_\gamma$.}
    as follows:
  \begin{enumerate}
  \item if $\Pi$ ends with a multicut, then $\Pi$ is $\gamma$-reducible if for every reduct $\Pi'$, $\Pi'$ is $\gamma$-reducible. 
  \item if $\Pi$ is
    $$\inferrule*[right=\impR]{{\deriv{\Pi'}{\Gamma,A\vdash B}}}{\Gamma\vdash A\imp B},$$
    then $\Pi$ is $\gamma$-reducible if $\lvl(\Pi)\leq \gamma$, $\Pi'$ is $\gamma$-reducible,
    and for every $\alpha$-reducible $\deriv{\Upsilon}{\Delta\vdash A}$ where $\alpha = \lvl(A)$, the derivation $\multcut(\Upsilon,\Pi')$ is $\gamma$-reducible.
  \item if $\Pi$ is $$\inferrule*[right=\impL]{\deriv{\Pi_1}{\Gamma\vdash A}\\ \deriv{\Pi_2}{\Gamma,B\vdash C}}{\Gamma,A\imp B\vdash C},$$ then $\Pi$ is $\gamma$-reducible if $\Pi_2$ is $\gamma$-reducible and $\Pi_1$ is normalizable.

  \item if $\Pi$ is $$\inferrule*[right=\IL]{\left\{\deriv{\Pi_{\vec u}}{B\ S\ \vec u\vdash S\ \vec u}\right\}_{\ground(\vec u)}\\ \deriv{\Pi'}{\Gamma,S\ \vec t\vdash C}}{\Gamma,p\ \vec t\vdash C}$$ then $\Pi$ is $\gamma$-reducible if $\Pi'$ is $\gamma$-reducible and each $\Pi_{\vec u}$ is normalizable.

  \item if $\Pi$ ends with any other rule, $\Pi$ is $\gamma$-reducible if each of its premise derivations $\Pi'$ is $\gamma$-reducible.
  \end{enumerate}
\end{definition}


We say a derivation $\Pi$ is \emph{reducible} if there exists a $\gamma$ such that $\Pi$ is $\gamma$-reducible. The lemma below asserts that reducibility implies normalizability. Since reducibility is a strengthening of normalizability, it is shown by straightforward induction on the $\gamma$-reducibility of a derivation.
\begin{restatable}[Normalization Lemma]{lemma}{lemnormalizationlemma}
  \label{normalization-lemma}
  If $\Xi$ is $\gamma$-reducible then $\Xi$ is normalizable.
\end{restatable}

\subsection{Reducibility lemma}
\label{sec:reducibility}

We now state and sketch a proof of the main lemma needed to prove normalizability of derivations in \LDINinf. We first define a new measure on derivations which is crucial for dealing with the inductive case of the lemma. The \emph{index} of a derivation $\Pi$ with premise derivations $\{\Pi_i\}_{i\in\mc I}$ is defined
\begin{mathpar}
  \Ind(\Pi) = \begin{cases}
    \sup_{i\in\mc I}\{\Ind(\Pi_i)\}+1\quad&\text{if $\Pi$ ends with \IL}\\
    \sup_{i\in\mc I}\{\Ind(\Pi_i)\}&\text{otherwise}
  \end{cases}
\end{mathpar}
this counts the maximum number of occurrences of \IL in a branch of a derivation. The main lemma is the following, whose detailed proof may be found in Appendix~\ref{appendix:reducibility-lemma}.

\begin{restatable}[Reducibility Lemma]{lemma}{lemreducibilitylemma}
  \label{reducibility-lemma}
  Given $n$ derivations $\Pi_1,\ldots,\Pi_n$ such that $\Pi_i$ is $\alpha_i$-reducible, and any derivation $\Pi$ the multicut
  \begin{mathpar}
    {\inferrule*[right=\multc]
      {\deriv{\Pi_1}{\Delta_1\vdash B_1}\\\ldots\\\deriv{\Pi_n}{\Delta_n\vdash B_n}\\
      \deriv{\Pi}{\Gamma,B_1,\ldots,B_n\vdash C}}
      {\Delta_1,\ldots,\Delta_n,\Gamma\vdash C}
    }
  \end{mathpar}
  which we denote by $\Xi$, is $\lvl(C)$-reducible.
\end{restatable}
\begin{proof}
  We proceed first by induction on the \textlabel{index $\Ind(\Pi)$}{ind-index}, then on the \textlabel{height $\height(\Pi)$}{ind-height}, then by subordinate induction on the \textlabel{number of cut formulas $n$}{ind-cut-formulas}.

  \case{Case $n = 0$} In this case $\Xi$ reduces to $\Pi$, so it suffices to show that $\Pi$ is reducible. We use the following trick: for any derivation $\Psi$ with height less than $\height(\Pi)$, induction hypothesis~\ref{ind-height} applies to the nullary multicut $\multcut(\Psi)$, and since this reduces to $\Psi$ it must be that $\Psi$ is reducible. It is then shown by case analysis on the last rule of $\Pi$ that $\Pi$ is reducible.
  
  \case{Case $n > 0$}
  We proceed by subordinate induction on the \textlabel{$\alpha_i$-reducibility of $\Pi_i$}{ind-reducibility} for $i = 1\ldots n$, and proceed by case analysis on the reduction of $\Xi$. We argue that this applies to each of the cases described in Section~\ref{sec:ldininf-cutred}.
  
  \subcase{\em \underline{Essential cases}} In each of these cases, the reduction decreases the height of the major premise derivation $\Pi$, allowing us to apply the induction hypothesis \ref{ind-height}. We show the \andR/\andL{i} case below as example:

    If $\Pi_1$ and $\Pi$ are
  \begin{mathpar}
    {\inferrule*[right=\andR]
      {\deduce{\Seq{\Delta_1}{B_1^1}}
        {\Pi_1^1}
        \\ \deduce{\Seq{\Delta_1}{B_1^2}}
           {\Pi_1^2}}
      {\Seq{\Delta_1}{B_1^1 \land B_1^2}}}\and
    {\inferrule*[right=\andL{i}]
      {\deduce{\Seq{B_1^i,B_2,\ldots,B_n,\Gamma}{C}}
        {\Pi'}}{\Seq{B_1^1 \land B_1^2,B_2,\ldots,B_n,\Gamma}{C}}
    }
  \end{mathpar}
  then the reduction $mc(\Pi_1^i,\Pi_2,\dots,\Pi_n,\Pi')$ is reducible from the induction hypothesis~\ref{ind-height} because $\Pi_1^i$ is reducible and $\height(\Pi')<\height(\Pi)$.

  \subcase{\em \underline{Inductive cases}}
  For this case to hold, we notice first that the reduction decreases the height of the major premise derivation $\Pi$, so by induction hypothesis~\ref{ind-height} it suffices to show that the unfolded derivation $\mu(\Pi_i,\Pi_S)$ is reducible whenever $\Pi_i$ is. This step depends on showing that for any reducible derivation $\Upsilon$, $\multc(\Upsilon,\init)$ is reducible, which crucially depends on induction hypothesis~\ref{ind-index}.
  
  \subcase{\em \underline{Left commutative cases}}
  In the case of a left commutative reduction, the multicut is pushed above the last rule of a minor premise $\Pi_i$. In particular, the new multicut involves a major premise $\Pi_i'$ of $\Pi_i$, which is reducible because $\Pi_i$ is reducible. Thus the new multicut $\multc(\Pi_1,\ldots,\Pi_i',\ldots,\Pi_n,\Pi)$ is reducible by the induction hypothesis~\ref{ind-reducibility}.

  \subcase{\em \underline{Right commutative cases}}
  In the case of a right commutative reduction, the multicut is pushed above the last rule of the major premise $\Pi$, so the height of the major premise is reduced. Thus the new multicut $\multc(\Pi_1,\ldots,\Pi_n,\Pi')$ is reducible by the induction hypothesis~\ref{ind-height}.
    
  \subcase{\em \underline{Structural cases}}
  In the case of a structural reduction, the height of the major premise is reduced. So even though the number of cut formulas may increase if $\Pi$ ends with a contraction, we may use induction hypothesis~\ref{ind-height}.
  
  \subcase{\em \underline{Left multicut case}}
  Similarly to the left commutative cases, the minor premise $\Pi_i$ reduces to a reducible derivation $\Pi_i'$, thus the reduct $\multc(\Pi_1,\ldots,\Pi_i',\ldots,\Pi_n,\Pi)$ is reducible by the induction hypothesis~\ref{ind-reducibility}.

  \subcase{\em \underline{Right multicut case}}
  In this case, the multicut in the major premise $\Pi$ is merged with the multicut in $\Xi$, but since the height of the major premise derivation decreases, induction hypothesis~\ref{ind-height} may be applied to obtain the reducibility of the reduct.
  
  \subcase{\em \underline{Left axiom case}}
  In this case, the number of cut formulas decreases, allowing us to apply induction hypothesis~\ref{ind-cut-formulas}.
  
  \subcase{\em \underline{Right axiom case}}
  In this case, the multicut $\multc(\Pi_1,\Pi)$ reduces to a nominal renaming $\Pi_1'$ of $\Pi_1$. Since renaming preserves reducibility, and $\Pi_1$ is reducible, the reduct is reducible.

\end{proof}

\subsection{Cut elimination}
\label{sec:cut-elimination}
We now proceed to show how the multicut rule may be eliminated from any \LDINinf derivation.
\begin{corollary}
  \label{reducibility-corollary}
  Every derivation $\Pi$ in \LDINinf is reducible.
\end{corollary}
\begin{proof}
  Consider the nullary multicut $\multcut(\Pi)$, which is reducible by Lemma~\ref{reducibility-lemma}. Since $\multcut(\Pi)$ reduces to $\Pi$ it follows that $\Pi$ is reducible.
\end{proof}

\begin{lemma}[Normal form lemma]
  \label{normal-form-lemma}
  If a derivation $\Pi$ of a sequent $\Gamma\vdash C$ is normalizable, then there exists a cut-free derivation $\hat\Pi$ of $\Gamma\vdash C$, which we call a \emph{normal form} for $\Pi$.
\end{lemma}
\begin{proof}
  By induction on normalizability. If $\Pi$ ends with a cut, first note that there necessarily exists a reduction $\Pi\to\Pi'$, and furthermore that $\Pi'$ is normalizable. By the induction hypothesis, $\Pi'$ has a cut-free normal form $\hat\Pi'$, which is also a cut-free normal form for $\Pi$.

  If $\Pi$ ends with any other rule $\circ$, then each of its premise derivations
  $$\left\{\deriv{\Pi_i}{\Gamma_i\vdash C_i}\right\}_{i\in\mc I}$$
  is normalizable, and therefore has a cut-free normal form $\{\hat\Pi_i\}_{i\in\mc I}$ of the same sequent. It follows that there exists a cut-free derivation $\hat\Pi$ obtained by applying $\circ$ to the premise derivations $\{\hat\Pi_i\}_{i\in\mc I}$.
\end{proof}

Since every derivation in \LDINinf is reducible by Corollary~\ref{reducibility-corollary}, reducibility implies normalizability by Lemma~\ref{normalization-lemma}, and a normalizable derivation has a cut-free normal form by the previous lemma, we obtain our main theorem
\begin{theorem}[Cut admissibility for \LDINinf]
  \label{main-theorem}
  Every derivation of a ground sequent $\Gamma\vdash C$ in \LDINinf admits a cut-free derivation of the same sequent.
\end{theorem}
\begin{proof}
  By the previous observation, any derivation making use of cut is normalizable, and therefore has a cut-free normal form by Lemma~\ref{normal-form-lemma}.
\end{proof}

\resetitemnum
}

\section{Applications}
\label{sec:applications}
{We now obtain the following important consequences of the cut elimination theorem.

\subsection{Consistency}
\begin{corollary}[Consistency of \LDINinf]
  \label{conldininf}
  \LDINinf is consistent.
\end{corollary}
\begin{proof}
  It suffices to notice by case analysis that there are no cut-free derivations of $\vdash \bot$. Since every derivation has a cut-free normal form by Theorem~\ref{main-theorem}, $\vdash\bot$ is not derivable in \LDINinf.
\end{proof}

\begin{corollary}[Consistency of \LDIN]
  \LDIN is consistent.
\end{corollary}
\begin{proof}
  This follows from the interpretation lemma (Lemma~\ref{conldin}) and the consistency of \LDINinf.
\end{proof}

\subsection{Proof search}
 An important consequence of the cut admissibility theorem is the following. A rule is \emph{invertible} (or \emph{asynchronous}) if the cut-free derivability of its conclusion implies the cut-free derivability of its premises. Furthermore, we say that a clause $H\defn B$ is invertible if $H$ is not unifiable with any other clause head $H'$
. The following corollary gives us completeness guarantees for proof search procedures in \LDINinf. 

\begin{corollary}[Inversion theorem for \LDINinf]
  \label{inversion-theorem}
   The logic \LDINinf satisfies the following meta-theoretic properties:
  \begin{enumerate}
  \item\label{inversion-theorem:1} The rules \andR, \andL*, \orL,\impR,\allR, \exL,\nabR,\nabL, \DL, 
    and \muR are invertible.
  \item\label{inversion-theorem:2} If a clause $H\defn B$ is invertible, then its corresponding \DR rule is invertible.
  \item\label{inversion-theorem:3} When the context is empty, we obtain the following partial invertability results for the rules \orR, \exR, and \DR:
    \begin{enumerate}
    \item If $\vdash A\vee B$ has a cut-free derivation, then either $\vdash A$ or $\vdash B$ have a cut-free derivation.
    \item If $\vdash \exists_\alpha x.C$ has a cut-free derivation, then $\vdash C[t/x]$ has a cut-free derivation for some $t\in\ground(\alpha)$.
    \item If $\vdash p\ \vec t$ has a cut-free derivation for a definition $p$, then there exists a clause $H\defn B$ and a formula $B'$ such that $\dfn(H\defn B,p\ \vec t,\epsilon_\emptyset,B')$ and $\vdash B'$ has a cut-free derivation.
  \end{enumerate}
\end{enumerate}
\end{corollary}
\begin{proof}
  This makes essential use of the cut admissibility property.

  For (\ref{inversion-theorem:1}), we only show the case of \andR, as all other cases are similar (except for \muR, which requires the intermediate theorem that $\Gamma,B\ p\ \vec t\vdash C$ implies $\Gamma,p\ \vec t\vdash C$, see Lemma~\ref{inductive-fixpoint}). Suppose that $\Gamma\vdash A\wedge B$ is derivable, then the following derivations may be obtained:
  \begin{mathpar}
    {\inferrule*[right=\multc]
      {\Gamma\vdash A\wedge B\qquad
        \inferrule*[Right=\andL1]
        {\inferrule*[Right=\init]
          { }{A\vdash A}}
        {A\wedge B\vdash A}}
      {\Gamma\vdash A}}\and
        {\inferrule*[right=\multc]
      {\Gamma\vdash A\wedge B\qquad
        \inferrule*[Right=\andL2]
        {\inferrule*[Right=\init]
          { }{B\vdash B}}
        {A\wedge B\vdash B}}
      {\Gamma\vdash B}}
  \end{mathpar}
  Since \multc can be eliminated, there exists cut-free proofs of $\Gamma\vdash A$ and $\Gamma\vdash B$.
  
  For (\ref{inversion-theorem:2}), first suppose that $A = H\rho$, and that $\Gamma\vdash A$ is derivable. We note that the only way for $\dfn(H'\defn B', A, \theta, B_0)$ to hold is if $A\theta = H'\rho'$ for some $\rho'$, and $B_0 = B'\rho'$. However this means $H'\rho' = A\theta = (H\rho)\theta$, which is not possible unless $H = H'$ (and $B = B'$) since $H\defn B$ is invertible. By Lemma~\ref{definition-uniqueness-lemma}, the $L_\lambda$ pattern condition on $H$ ensures that $\rho'$ is unique, and thus $\rho' = \theta\circ\rho$, from which it follows that $B_0 = (B'\rho') = (B\rho)\theta$, and thus we obtain the derivation
  \begin{mathpar}
    \inferrule*[right=\multc]
    {\Gamma\vdash A\\
      {\inferrule*[Right=\DL, rightskip=8em]
      {\left\{\inferrule*[right=\init]{ }{B_0 \vdash (B\rho)\theta}\st
          \dfn(H'\defn B', A, \theta, B_0)\right\}}
      {A\vdash B\rho}
    }}
  {\Gamma\vdash B\rho}
\end{mathpar}
Since cut may be eliminated, it follows that $\Gamma\vdash B\rho$ has a cut free derivation.

  For (\ref{inversion-theorem:3}), this follows by case analysis on the last rule of the derivation, which can only be a right introduction rule. This does not require the cut admissibility property, and as such also holds in \LDIN.
\end{proof}

We note that this theorem only applies to the infinitary logic \LDINinf, and not to \LDIN. Extending this theorem to \LDIN requires a cut elimination theorem for \LDIN, which is presently out of reach.
}

\section{Related and future work}
\label{sec:conclusion}
{
In this work, we've shown how to extend the logic \LD to a new logic \LDIN with inductive definitions and nabla quantification, which we've proven to be consistent. The addition of weak stratification allows us to make definitions commonly used to reason about computer programs, such as logical relations. The eventual aim of this work is to extend the logic \G, underlying Abella, with this broader body of definitions.

In another line of work towards this goal, Baelde and Nadathur have developed the \NJI calculus \cite{BaNa12}, which is an extension of deduction modulo (see eg. Dowek \cite{Do14}) with inductive and co-inductive definitions. Their work similarly shows that the same kind of definitions are permissible, but using a technique involving a rewriting system on propositions. Consistency is obtained by proving strong normalization for \NJI.

However, there are several future extensions to this work that seem more promising from the direction of \LD. First, although \NJI includes co-induction, it is not clear how to define nominal abstraction in the natural deduction context, whereas we believe that adding co-induction and nominal abstraction to \LDIN will not add much complexity to our present work. Adding these features into the logic is a necessary step to fully justify extending the Abella proof assistant. 

Second, we desire a cut elimination theorem for our target extension of \G. We saw with Corollary~\ref{inversion-theorem} that a cut elimination theorem allows us to obtain certain desirable properties about a logic. However, we are presently only able to obtain cut elimination for the ground logic \LDINinf, and not the logic \LDIN. Similarly, \NJI does not help us in this regard, since the structure of natural deduction derivations is quite different from that of the sequent calculus.

Third, we envision extending the logic to allow for higher order quantification. This would on the one hand allow for much greater generality in the statement of theorems and definitions, and on the other, the increased strength of our logic would allow us to prove consistency theorems for higher order logics. In this regard, the \NJI calculus's introduction of rewriting into the system affects the intensional treatment of syntax that is crucial to preserving our desired interpretation of equality. This is not an immediate concern in a first order system for which rewriting is restricted to propositions. However, this leads to further complications for a higher order system.

In addition to these key research directions, we envision two other possibilities for future work. First, the logics \LD and \LDIN both require restricting definition clauses $H\defn_{\mc X} B$ to those in which $H$ must lie in the $L_\lambda$ pattern fragment. We believe this is not an essential restriction, and that our current proof can be adapted by a slight modification of the ground \DL rule to remove this restriction. Finally, we aspire to eventually formalize this work. One could imagine formalizing this work inside the logic we are defining itself, which, in order to define reducibility as we are using it would require a notion of induction compatible with weak stratification. Note that \`a priori this poses no G\"odelian circularities since the logic of definitions is not a single logic, but a family of logics indexed by a set of definitions whose stratification must be proven at the meta-theoretic level.



\section{Aknowledgements}
The idea of expanding upon Tiu's work was motivated and proposed by my advisor, Gopalan Nadathur, who's advising, through patient explanation and many discussions, has been invaluable to the development of this work. In addition, the present document has been thoroughly reviewed by my advisor, whose many suggestions have greatly improved this work. Finally, the present counter-example in section~\ref{examples:ind-counterexample} is due to Nadathur, who found a way to eliminate an assumption about fixed-points from a previous example.
}

\bibliographystyle{abbrv}
\bibliography{biblio}

\appendix
{
\section{Proofs of lemmas}
\label{appendix:lemmas}
\subsection{The renaming lemma}
\label{appendix:renaming-lemma}

Denote $A\approx A'$ if there exists a permutation $\pi\in\Perm$ such that $A' = A[\pi]$. Note that this is an equivalence relation. More generally, say $A_1,\ldots,A_n\approx A_1',\ldots,A_n'$ if $A_i \approx A_i'$ for each $i$. We now define $\Pi\sim \Pi'$ if both of the following conditions are satisfied
\begin{enumerate}
\item $\Pi$ and $\Pi'$ each end with sequents $\Gamma \vdash C$ and $\Gamma' \vdash C'$ such that  $\Gamma\approx \Gamma'$ and $C \approx C'$.
\item $\Pi$ and $\Pi'$ end with the same rule and $\Pi_i\sim\Pi_i'$ for all the premise derivations $\{\Pi_i\}_{i\in\mc I}$ of $\Pi$ and $\{\Pi_i'\}_{i\in\mc I}$ of $\Pi'$.
\end{enumerate}
We note in particular that if $\Pi\sim\Pi'$, then $\height(\Pi) = \height(\Pi')$. Furthermore, $\sim$ is an equivalence relation.

\begin{lemma}[Renaming Lemma]
  \label{renaming-lemma}
Let $\Pi$ be a derivation of $B_1,\ldots, B_n\vdash C$, and suppose $B_i\approx B_i'$ and $C\approx C'$, then there exists a derivation $\Pi'$ of $B_1',\ldots,B_n'\vdash C'$ such that $\Pi\sim\Pi'$.
\end{lemma}
\begin{proof}
  By induction on the derivation $\Pi$. We only do the cases for \multc, \axiom, and \nabR.
  \case{\axiom} Suppose $\Pi$ ends with \axiom, in particular $C \approx B_i$ for some $i$. Then by transitivity, $C' \approx C\approx B_i\approx B_i'$, thus $\Pi'$ is again \axiom.

  \case{\multc} Suppose $\Pi$ ends with \multc, with cut formulas $B_1,\ldots,B_n$ and premise derivations $\Pi_1,\ldots,\Pi_n$ and $\Pi$. Then by the induction hypothesis, there are derivations $\Pi_i'$ of $\Delta'\vdash B_i$ and $\Pi'$ of $\Gamma',B_1,\ldots, B_n\vdash C'$ since $\approx$ is reflexive. So $\Pi'$ is \multc with those premise derivations and the same cut formulas.

  \case{\nabR} Suppose $\Pi$ ends with \nabR, introducing the formula $\nabla x.C$. Then the premise derivation $\Pi_1$ is a derivation of $\Gamma\vdash C[n/x]$ where $n\not\in\supp(C)$. Suppose $n'\not\in\supp(C')$, then by the induction hypothesis there exists a derivation $\Pi_1'$ of $\Gamma'\vdash C'[n'/x]$, and $\Pi'$ is thus again \nabR.
\end{proof}

The next lemma shows that the renaming relation is a congruence with respect to height, levels, reduction, normalizability, and reducibility. 
\begin{lemma}[Renaming congruence lemma]
  \label{renaming-congruence-lemma}
  We make the following observations about the renaming relation $\Xi\sim\Xi'$:
  \begin{enumerate}
  \item If $\Pi\sim\Pi'$, then $\height(\Pi)=\height(\Pi')$ and $\lvl(\Pi) = \lvl(\Pi')$.
  \item If $\Pi\to\Xi$, then for any $\Pi'$ such that $\Pi\sim\Pi'$, there exists $\Xi'$ such that $\Pi'\to\Xi'$ and $\Xi\sim\Xi'$.
  \item If $\Pi$ is normalizable and $\Pi\sim\Pi'$, then $\Pi'$ is normalizable.
    
  \item If $\Pi$ is $\gamma$-reducible and $\Pi\sim\Pi'$, then $\Pi'$ is $\gamma$-reducible.
  \end{enumerate}
\end{lemma}
\begin{proof}
  We omit the details of this proof. Part (1) is clear by induction on the height of the derivation. Part (2) is shown by induction on the height of the derivation and by case analysis on the reduction relation. Part (3) is shown by induction on normalizability, together with the result from part (2). Part (4) is shown first by transfinite induction on the level $\gamma$ and then by subordinate induction on $\gamma$-reducibility, together with the results from parts (2) and (3).
\end{proof}

\subsection{The grounding lemma}
\label{appendix:grounding-lemma}
\lemgroundinglemma*
\begin{proof}
  The proof is by induction on the height of the derivation of $\mc X;\Gamma\vdash C$. Note that we do not require $\supp(\delta)$ to be empty. This is straightforward for all rules except the following (we omit the symmetric cases \exL, \exR, and \nabL):
  \case{\allR} Suppose the derivation ends with \allR with $C = \forall_\alpha x.D$.
  Given $y\not\in\mc Y$ and $\vec n = \supp(D)$, suppose $\delta$ is an $\mc Y$-grounding substitution. Since $x$ may be renamed in $\forall_\alpha x.D$, we will assume that $x\not\in \mc Y$. Now for every $t\in\ground(\alpha)$, notice that
  $$(D[y\ \vec n/x])[\delta,(\lambda\vec x.t)/y] = (D[\delta,y\ \vec n/x])[\lambda\vec x.t/y] = (D[\delta,x/x])[t/x].$$

  Since $\delta' = [\delta,(\lambda\vec x.t)/y]$ is an $\mc Y,y$-grounding substitution, we may apply it to the premise ${\mc Y,y;\Gamma\vdash D[y\ \vec n/x]}$ to obtain a ground derivation of $\Gamma\delta'\vdash (D[\delta,x/x])[t/x]$ by the induction hypothesis, as follows
    \begin{mathpar}
    {\inferrule*[right=\allR]
      {\mc Y,y;\Gamma\vdash D[y\ \vec n/x]}{\mc Y;\Gamma\vdash \forall_\alpha x.D}}
    \quad\leadsto\quad
    {\inferrule*[right=\allR]
      {\left\{\Gamma\delta'\vdash (D[\delta,x/x])[t/x])\right\}_{t\in\ground(\alpha)}}
      {\Gamma\delta\vdash (\forall_\alpha x.D)\delta}}
  \end{mathpar}
  where we note that $(\forall_\alpha x.D)\delta = \forall_\alpha x.(D[\delta,x/x])$.
   \case{\allL}
   Suppose the derivation ends with \allL with $\Gamma = \Gamma', \forall_\alpha x.B$,
   and suppose $t : \alpha$ is a term in context $\mc Y$. Since $x$ may be renamed in $\forall x.B$, we will assume that $x\not\in\mc Y$.
   We apply the induction hypothesis to the premise with the ground substitution $\delta$, and since
   $$(B[t/x])\delta = B[\delta,t\delta/x] = (B[\delta,x/x])[t\delta/x],$$
   we obtain a ground derivation of $\Gamma\delta, (\forall_\alpha x.B)\delta \vdash C\delta$ as follows:
  \begin{mathpar}
    {\inferrule*[right=\allL]
      {\mc Y;\Gamma', B[t/x]\vdash C\\
      \mc Y\vdash t : \alpha}
      {\mc Y;\Gamma',\forall_\alpha x.B\vdash C}}\quad\leadsto\quad
    {\inferrule*[right=\allL]
      {\Gamma'\delta, (B[\delta,x/x])[t\delta/x]\vdash C\delta\\
      \vdash t\delta : \alpha}
      {\Gamma'\delta,(\forall_\alpha x.B)\delta\vdash C\delta}}
  \end{mathpar}
  where we note that $(\forall_\alpha x.B)\delta = \forall_\alpha x.(B[\delta,x/x])$.

  \case{\DL} Suppose the derivation ends with \DL on the atom $A$ where $\Gamma = \Gamma', A$. It is clear that if $\dfn(H\defn_{\mc X} B, A\delta, \epsilon_\emptyset, B')$, then $\dfn(H\defn_{\mc X} B, A, \delta, B')$, and furthermore that $\range(\delta) = \emptyset$ since $\delta$ is a ground substitution. Thus an \LDINinf derivation of a premise $\Gamma'\delta, B'\vdash C\delta$ below is obtained by applying the induction hypothesis to the ground premise $\emptyset;\Gamma'\delta,B'\vdash C\delta$ with the empty grounding substitution. 

    \begin{mathpar}
    {\inferrule*[right=\DL]
      {\{\mc Z; \Gamma'\theta, B'\vdash C\theta\st \dfn(H\defn_{\mc X} B, A, \theta, B')
        \text{ and }\mc Z = \range(\theta)\}}
      {\mc Y; \Gamma', A \vdash C}
    }\quad\leadsto\quad
    {\inferrule*[right=\DL]
      {\{\Gamma'\delta, B' \vdash C\delta\st\dfn(H\defn_{\mc X} B, A\delta, \epsilon_\emptyset,B')\}}
      {\Gamma'\delta, A\delta \vdash C\delta}
    }
    \end{mathpar}
We note that this has the effect of pruning branches from \DL.

  \case{\DR} Suppose the derivation ends with \DR with $C = A$ for an atom $A$. Note that if $A = H\rho$, then $A\delta = (H\rho)\delta$, thus $\dfn(H\defn B,A,\epsilon_\emptyset,B')$ implies $\dfn(H\defn B,A\delta,\epsilon_\emptyset,B'\delta)$, thus by applying the induction hypothesis to the premise $\mc Y;\Gamma\vdash B'$ with the ground substitution $\delta$, we obtain
  \begin{mathpar}
    {\inferrule*[right=\DR]
      {\mc Y;\Gamma\vdash B'}
      {\mc Y;\Gamma\vdash A}}
    \quad\leadsto\quad
    {\inferrule*[right=\DR]
      {\Gamma\delta\vdash B'\delta}
      {\Gamma\delta\vdash A\delta}}
  \end{mathpar}
  

  \case{\nabR} Suppose the derivation ends with \nabR with $C = \nabla_\alpha x. D$. Since $x$ may be renamed in $\nabla_\alpha x.D$, we may assume that $x\not\in\mc Y$. Now since $\delta$ may have nonempty support, we first pick $m\not\in\supp(D[\delta,x/x])$. Then, we apply the renaming lemma (\ref{renaming-lemma}) to the premise $\mc Y;\Gamma\vdash D[n/x]$ to obtain a derivation of ${\mc Y;\Gamma[m/n]\vdash D[m/x]}$ of the same height as the original premise derivation. We can thus apply it to the induction hypothesis with the ground substitution $\delta$ to obtain the ground premise $\Gamma[m/n]\vdash D[\delta,x/x][m/x]$ to obtain
  \begin{mathpar}
    {\inferrule*[right={\nabR\quad$n\not\in\supp(D)$}]
      {\mc Y;\Gamma\vdash D[n/x]}
      {\mc Y;\Gamma\vdash \nabla_\alpha x.D}}
    \quad\leadsto\quad
    {\inferrule*[right={\nabR\quad$m\not\in\supp(D[\delta,x/x])$}]
      {\Gamma[m/n]\delta\vdash (D[\delta,x/x])[m/x]}
      {\Gamma[m/n]\delta\vdash (\nabla_\alpha x.D)\delta}}
  \end{mathpar}
  where we note $(\nabla_\alpha x.D)\delta = \nabla_\alpha x.(D[\delta,x/x])$. A second application of the renaming lemma (\ref{renaming-lemma}) gives us a ground derivation of $\Gamma\delta\vdash(\nabla_\alpha x.D)\delta$.

  \case{\IL} Suppose the derivation ends with \muL with $\Gamma = \Gamma', p\ \vec t$ for an inductive definition $p\ \vec x \ind B\ p\ \vec x$. Here we must first apply the induction hypothesis to the derivation of $\vec x; B\ S\ \vec x\vdash S\ \vec x$ with the $\vec x$-grounding substitution $[\vec t/\vec x]$ for every $\vec t\in\ground(\vec\alpha)$, and then apply the induction hypothesis to the premise $\mc X;\Gamma,S\ \vec t\vdash C$ with the ground substitution $\delta$ to obtain
  \begin{mathpar}
    {\inferrule*[right=\IL]
      {{\vec x; B\ S\ \vec x\vdash S\ \vec x}\\
        {\mc X; \Gamma',S\ \vec t\vdash C}}
      {{\mc X;\Gamma', p\ \vec t\vdash C}}}
    \quad\leadsto\quad
    {\inferrule*[right=\IL]
      {\left\{B\ S\ \vec t\vdash S\ \vec t\right\}_{\vec t\in\ground(\vec\alpha)}\\
        \Gamma'\delta,S\ \vec t\delta\vdash C\delta}
      {\Gamma'\delta,(p\ \vec t)\delta\vdash C\delta}}
  \end{mathpar}

\end{proof}

\subsection{The unfolding lemma}
\label{appendix:unfolding-lemma}
We first make a technical definition. Given an inference rule, if it is a right introduction rule, all of its premises are called \emph{major} premises. If it is a left introduction rule or a multicut, then only those with the same consequent as its conclusion are major premises. All other premises are \emph{minor} premises. The only rules with minor premises are \impL, \IL, and \multcut, which will often require special attention.

\begin{lemma}
    Suppose $p\ \vec x\ind B\ p\ \vec x$ is an inductive definition, then for any derivation $\Xi$ of $\Delta\vdash C\ p$ where $p$ does not occur in $C$ and occurs only positively in $C\ p$ (ie. does not occur to the left of an implication), there exists a derivation $\mu(\Xi,\Pi_S)$ of $\Delta \vdash C\ S$.
\end{lemma}
\begin{proof}
  By induction on the height of the derivation $\Xi$. If $C\ p = C\ (\lambda \vec x.\top)$ (that is $C\ p$ does not depend on $p$) then $\mu(\Xi,\Pi_S) := \Xi$. If $\Xi$ ends with any rule (including a multicut) except \axiom, \impR, or \IR, then $\mu(\Xi,\Pi_S)$ is obtained by unfolding each major premise derivation. For example, 
  \case{\impL} If $\Xi$ ends with \impL, then $\mu(\Xi,\Pi_S)$ is obtained as follows
  \begin{mathpar}
    {\inferrule*[right=\impL]
      {\deriv{\Xi_1}{\Delta\vdash D_1}\\
        \deriv{\Xi_2}{\Delta,D_2\vdash C\ p}}
      {\Delta, D_1\imp D_2\vdash C\ p}}\qquad\leadsto\qquad
    {\inferrule*[right=\impL]
      {\deriv{\Xi_1}{\Delta\vdash D_1}\\
        \deriv{\mu(\Xi_2,\Pi_S)}{\Delta,D_2\vdash C\ S}}
      {\Delta,D_1\imp D_2\vdash C\ S}}
  \end{mathpar}

 We must treat specially the case in which $\Xi$ ends with \axiom, \impR, or \IR.

  \case{\axiom} If $\Xi$ ends with \axiom on the fixpoint $p$, then $\mu(\Xi,\Pi_S)$ is obtained as follows
  \begin{mathpar}
    {\inferrule*[right=\axiom]{ }
               {p\ \vec t\vdash (p\ \vec t)[\pi]}}\qquad\leadsto\qquad
    {\inferrule*[right=\IL]
      {\left\{\deriv{\Pi_S^{\vec u}}{B\ S\ \vec u\vdash S\ \vec u}\right\}_{\vec u\in\ground(\alpha)}\\
        {\raisebox{-.8em}{$\inferrule*[Right=\init]{ }{S\ \vec t\vdash (S\ \vec t)[\pi]}$}}}
      {p\ \vec t \vdash (S\ \vec t)[\pi]}}
  \end{mathpar}

  \case{\impR} If $\Xi$ ends with \impR, then $\mu(\Xi,\Pi_S)$ is obtained as follows
  \begin{mathpar}
    {\inferrule*[right=\impR]
      {{\deriv{\Xi'}{\Delta, C_1\vdash C_2\ p}}}
      {\Delta\vdash C_1\imp C_2\ p}}\qquad\leadsto\qquad
    {\inferrule*[right=\impR]
      {{\deriv{\mu(\Xi',\Pi_S)}{\Delta,C_1\vdash C_2\ S}}}
      {\Delta\vdash C_1\imp C_2\ S}}
  \end{mathpar}
  where we note that since $p$ occurs only positively in $C\ p$, $p$ may not appear in $C_1$.



  \case{\IR} If $\Xi$ ends with \IR for the fixpoint $p$, then $\mu(\Xi,\Pi_S)$ is obtained as follows
  \begin{mathpar}
    {\inferrule*[right=\IR]
      {{\deriv{\Xi'}{\Delta\vdash B\ p\ \vec t}}}
      {\Delta\vdash p\ \vec t}}\qquad\leadsto\qquad
    {\inferrule*[right=\multc]
      {\deriv{\mu(\Xi',\Pi_S)}{\Delta\vdash B\ S\ \vec t}\\
        \deriv{\Pi_S^{\vec t}}{B\ S\ \vec t\vdash S\ \vec t}}
      {\Delta\vdash S\ \vec t}}
  \end{mathpar}
  Note here that we must use the constraint on inductive definitions being (strictly) stratified which guarantees that $p$ occurs only positively in $B\ p\ \vec t$, allowing us to apply the inductive hypothesis.


\end{proof}

\subsection{The normalization lemma}
\label{appendix:normalization-lemma}

\lemnormalizationlemma*
\begin{proof}
  We proceed by induction on the $\gamma$-reducibility of $\Xi$.
  
  \case{Case 1} Suppose $\Xi$ is a multicut with $\gamma$-reducible reduct $\Xi'$, then by the induction hypothesis $\Xi'$ is normalizable, and therefore $\Xi$ is normalizable.
  
  \case{Case 2} Suppose $\Xi$ ends with \impR introducing the formula $A\imp B$. By definition, $\Xi$ has $\gamma$-reducible premise derivation $\Pi'$, so by the induction hypothesis, $\Pi'$ is normalizable, and thus $\Xi$ is normalizable.

  \case{Case 3} Suppose $\Xi$ ends with \impL with premise derivations $\Pi_1$ and $\Pi_2$. By definition, $\Pi_1$ is already normalizable. Similarly, $\Pi_2$ is $\gamma$-reducible, thus $\Pi_2$ is normalizable by the induction hypothesis, therefore $\Xi$ is normalizable.

  \case{Case 4} Suppose $\Xi$ ends with \IL with minor premise derivations $\{\Pi_{\vec u}\}_{\ground(\vec u)}$ and the major premise derivation $\Pi'$, then the $\Pi_{\vec u}$ are already nomalizable and $\Pi'$ is $\gamma$-reducible. By the induction hypothesis it follows that $\Pi'$ is normalizable, therefore $\Xi$ is normalizable.

  \case{Case 5} Suppose $\Xi$ ends with any other rule, then each of its premises $\Pi_i$ is a major premise which is $\gamma$-reducible, and therefore normalizable by the induction hypothesis. It follows that $\Xi$ is normalizable.
\end{proof}

\subsection{Reducibility of unfolded derivations}
  \label{appendix:reducibility-of-unfolded-derivations}

 The following lemma shows that the unfolding operation $\mu(\Pi,\Pi_S)$ commutes with the reduction relation. The results in this section follow closely those of \cite{Tiu04}.
  
 \begin{lemma}[Unfolding commutation lemma]
   \label{unfolding-commutation-lemma}
    Suppose $\Pi^{\vec u}_S$ is a family of derivations of $B\ S\ \vec u\vdash S\ \vec u$ for an inductively defined predicate $p$, then for any derivations $\Pi_1,\ldots, \Pi_n$, and $\Pi$, where $\Pi$ does not end with \axiom or \muR with consequent $p\ \vec t$, we have $\mu(\multcut(\Pi_1,\ldots,\Pi_n,\Pi),\Pi_S) \to \Psi$ if and only if there exists $\Psi'$ such that $\Psi \sim \mu(\Psi',\Pi_S)$ and $\multcut(\Pi_1,\ldots,\Pi_n,\Pi)\to\Psi'$.
\end{lemma}

\begin{proof}
    This must be done by case analysis on the cut reduction $\mu(\multcut(\Pi_1,\ldots,\Pi_n,\Pi),\Pi_S)\to \Psi$. Note that unfolding is defined to be simply the unfolding of major premise derivations except for the cases in which $\Pi$ ends with \muR or \axiom with consequent $p\ \vec t$. This observation means the lemma may be checked in a very mechanical manner for every case except the right commutative and right axiom cases (-/\muR) and (-/\axiom).
\end{proof}

The next lemma characterizes sufficient conditions for unfolded derivations to be reducible.
\begin{lemma}[Reducibility of unfolded derivations]
  \label{unfolded-reducibility-lemma}
Given an inductive definition $p\ \vec x \ind B\ p\ \vec x$, an inductive invariant $S$ of the same type as $p$, and a family of derivations $\Pi^{\vec u}_S$ of $B\ S\ \vec u\vdash S\ \vec u$ for any ground $\vec u$, suppose
\begin{enumerate}
\item \label{unfold-red:ass:1}Each derivation $\Pi^{\vec u}_S$ is normalizable.
\item \label{unfold-red:ass:2}For any reducible derivation $\Upsilon$, the derivation
\begin{mathpar}
  {\inferrule*[right=\multc]
    {\deriv{\Upsilon}{\Delta\vdash B\ S\ \vec t}\\
      \deriv{\Pi_S^{\vec t}}{B\ S\ \vec t\vdash S\ \vec t}}
    {\Delta\vdash S\ \vec t}
  }
\end{mathpar}
is reducible.
\item \label{unfold-red:ass:3}For any formula $C$ and any reducible $\Upsilon$, the derivation
  \begin{mathpar}
    \inferrule*[right=\multc]
    {\deriv{\Upsilon}{\Delta\vdash C}\\
      {\raisebox{-.6em}{$\inferrule*[Right=\init]{ }{C\vdash C[\pi]}$}}}
    {\Delta\vdash C[\pi]}
  \end{mathpar}
is reducible.
\end{enumerate}
Then for any $\alpha$-reducible derivation $\Xi$ in which $p$ occurs only positively in its conclusion, the unfolding $\mu(\Xi,\Pi_S)$ is reducible.
\end{lemma}
\begin{proof}
  
  We proceed first by transfinite induction on $\alpha$, then by induction on the $\alpha$-reducibility of $\Xi$, and finally by case analysis on the last rule of $\Xi$. If $\Xi$ ends with any rule except \axiom, \impR, \muR, or \multc, recall that $\mu(\Xi,\Pi_S)$ is obtained by unfolding each major premise derivation. For example,
  \case{\impL} If $\Xi$ ends with \impL with premise derivations $\Xi_1$ and $\Xi_2$, then $\mu(\Xi,\Pi_S)$ is
  \begin{mathpar}
    \inferrule*[right=\impL]
    {\deriv{\Xi_1}{\Delta\vdash D_1}\\
      \deriv{\mu(\Xi_2,\Pi_S)}{\Delta, D_2\vdash C\ S}}
    {\Delta,D_1\imp D_2\vdash C\ S}
  \end{mathpar}
  which is $\beta$-reducible if $\mu(\Xi_2,\Pi_S)$ is $\beta$-reducible and $\Xi_1$ is normalizable. Since $\Xi$ is $\alpha$-reducible it follows that $\Xi_1$ is normalizable and $\Xi_2$ is $\alpha$-reducible, so by the second induction hypothesis it follows that $\mu(\Xi_2,\Pi_S)$ is $\beta$-reducible, as desired. More generally, $\beta$-reducibility of the major premise derivations follows by the second inductive hypotheses, and normalizability of the minor premise derivations follows from the $\alpha$-reducibility of $\Xi$.

  We now treat specially the cases \axiom, \impR, \muR, and \multc.
  
  \case{\axiom} Suppose $\Xi$ is \axiom, then $\mu(\Xi,\Pi_S)$ is
  \begin{mathpar}
    \inferrule*[right=\IL]
    {\left\{\deriv{\Pi_S^{\vec u}}{B\ S\ \vec u\vdash S\ \vec u}\right\}_{\vec u\in\ground(\alpha)}\\
      {\raisebox{-.8em}{$\inferrule*[Right=\init]
      { }{S\ \vec t\vdash(S\ \vec t)[\pi]}$}}}
    {p\ \vec t\vdash (S\ \vec t)[\pi]}
  \end{mathpar}
  which is $\beta$-reducible because the $\Pi_S^{\vec u}$ are normalizable by assumption~(\ref{unfold-red:ass:1}), and \init is $\alpha$-reducible because it is cut-free.
  
  \case{\impR} Suppose $\Xi$ ends with the \impR rule with premise derivation $\Xi'$ and end sequent $\Gamma\vdash C_1 \imp C_2\ p$, then $\mu(\Xi,\Pi_S)$ is
  \begin{mathpar}
    \inferrule*[right=\impR]
    {{\deriv{\mu(\Xi',\Pi_S)}{\Gamma,C_1 \vdash C_2\ p}}}
    {\Gamma\vdash C_1\imp C_2\ p}
  \end{mathpar}
  Now suppose that $\alpha_1 = \lvl(C_1)$ and $\alpha_2 = \lvl(C_2\ p)$, then by the $\alpha$-reducibility of $\Xi$, we know that $\Xi'$ is $\alpha_2$-reducible. So if $\alpha_2 = \alpha$, then $\mu(\Xi',\Pi_S)$ is reducible by the second inductive hypothesis. If $\alpha_2 < \alpha$ then $\mu(\Xi',\Pi_S)$ is reducible by the first induction hypothesis. Finally, we know that for any $\alpha_1$-reducible $\Upsilon$, the derivation
  \begin{mathpar}
    \inferrule*[right=\multc]
    {\deriv{\Upsilon}{\Delta\vdash C_1}\\\deriv{\Xi'}{\Gamma,C_1\vdash C_2\ p}}
    {\Gamma,\Delta\vdash C_2\ p}
  \end{mathpar}
  is $\alpha_2$-reducible. We wish to show that $\multcut(\Upsilon,\mu(\Xi',\Pi_S))$ is reducible. If $\alpha_2 < \alpha$, then applying the first induction hypothesis tells us that $\mu(\multcut(\Upsilon,\Xi'),\Pi_S) = \multcut(\Upsilon,\mu(\Xi',\Pi_S))$ is reducible. If $\alpha_2 = \alpha$, then it follows from the second induction hypothesis that $\mu(\multcut(\Upsilon,\Xi'),\Pi_S) = \multcut(\Upsilon,\mu(\Xi',\Pi_S))$ is reducible.

  \case{\muR} If $\Xi$ ends with \muR with premise derivation $\Xi'$, then $\mu(\Xi,\Pi_S)$ is
  \begin{mathpar}
    \inferrule*[right=\multc]
    {\deriv{\mu(\Xi',\Pi_S)}{\Gamma\vdash B\ S\ \vec t}\\
      \deriv{\Pi_S^t}{B\ S\ \vec t\vdash S\ \vec t}}
    {\Gamma\vdash S\ \vec t}
  \end{mathpar}
  which is reducible from assumption~(\ref{unfold-red:ass:2}) of the lemma provided $\mu(\Xi',\Pi_S)$ is reducible. If $\lvl(B\ S\ \vec t) < \lvl(S\ \vec t)$, then this follows from the first inductive hypothesis, otherwise $\lvl(B\ S\ \vec t) = \lvl(S\ \vec t)$ and the result follows from the second inductive hypothesis.
  
  \case{\multc} Suppose $\Xi$ is a multicut with minor premises $\Pi_1,\ldots, \Pi_n$ and major premise $\Xi'$ with end sequent $\Delta_1,\ldots,\Delta_n,\Gamma\vdash C\ p$. We first suppose that the major premise $\Xi'$ does not end with \axiom or \muR. Now suppose that $\mu(\Xi,\Pi_S) \to\Psi$, then by Lemma~\ref{unfolding-commutation-lemma}, there exists $\Psi'$ such that $\Psi \sim \mu(\Psi',\Pi_S)$, and $\Xi\to\Psi'$. In particular, $\mu(\Psi',\Pi_S)$ is reducible by the induction hypothesis which implies that $\mu(\Xi,\Pi_S)$ is reducible.

  We must now consider the remaining cases in which $\Xi'$ ends with \axiom or \muR, in which case $C\ p = (p\ \vec t)[\pi]$.

  \subcase{\underline{If $\Xi'$ ends with \axiom}} Suppose $\Xi$ is
  \begin{mathpar}
    \inferrule*[right=\multc]
    {\deriv{\Pi_1}{\Delta_1\vdash p\ \vec t}\\
      \raisebox{-.7em}{$\inferrule*[Right=\axiom]
      { }{p\ \vec t\vdash (p\ \vec t)[\pi]}$}}
    {\Delta_1\vdash (p\ \vec t)[\pi]}
  \end{mathpar}
  which reduces to a derivation $\Pi_1'$ with $\Pi_1'\sim\Pi_1$. By assumption, $\Pi_1'$ (and thus $\Pi_1$ by Lemma~\ref{renaming-congruence-lemma}) is $\alpha$-reducible, and the second induction hypothesis implies $\mu(\Pi_1,\Pi_S)$ is reducible. Now $\mu(\Xi,\Pi_S)$ is the following
  \begin{mathpar}
    \inferrule*[right=\multc]
    {\raisebox{.5em}{$\deriv{\Pi_1}{\Delta_1\vdash p\ \vec t}$}\\
      \inferrule*[Right=\muL,rightskip=7em]
      {\left\{\deriv{\Pi_S^{\vec u}}{B\ S\ \vec u\vdash S\ \vec u}\right\}_{\vec u\in\ground(\alpha)}\\
        {\raisebox{-1em}{$\inferrule*[Right=\init]
        { }{S\ \vec t\vdash (S\ \vec t)[\pi]}$}}}
      {p\ \vec t\vdash (S\ \vec t)[\pi]}}
    {\Delta\vdash (S\ \vec t)[\pi]}
  \end{mathpar}

  which itself reduces to
  \begin{mathpar}
    \inferrule*[right=\multc]
    {\deriv{\mu(\Pi_1,\Pi_S)}{\Delta_1\vdash S\ \vec t}\\
      \raisebox{-.7em}{$\inferrule*[Right=\init]
      { }{S\ \vec t\vdash (S\ \vec t)[\pi]}$}}
    {\Delta_1\vdash (S\ \vec t)[\pi]}
  \end{mathpar}
  which is reducible by assumption~(\ref{unfold-red:ass:3}).

  \subcase{\underline{If $\Xi'$ ends with \muR}} Suppose $\Xi$ is
  \begin{mathpar}
    \inferrule*[right=\multc]
    {{\deriv{\Pi_1}{\Delta_1\vdash B_1}}\\\ldots\\
      {\deriv{\Pi_n}{\Delta_n\vdash B_n}}\\
      {\raisebox{-.7em}{$\inferrule*[Right=\muR]
      {{\deriv{\Xi'}{\Gamma,B_1,\ldots,B_n\vdash B\ p\ \vec t}}}
      {\Gamma,B_1,\ldots,B_n\vdash p\ \vec t}$}}}
    {\Delta_1,\ldots,\Delta_n,\Gamma\vdash p\ \vec t}
  \end{mathpar}
  which reduces to $\text{\muR}(\multcut(\Pi_1,\ldots,\Pi_n,\Xi'))$. Since $\Xi$ is $\alpha$-reducible, it must be that the premise derivation $\multcut(\Pi_1,\ldots,\Pi_n,\Xi')$ is $\alpha$-reducible. Now the unfolding $\mu(\Xi,\Pi_S)$ is the multicut
  $$\multcut(\Pi_1,\dots,\Pi_n,\mu(\text{\muR}(\Xi'),\Pi_S)) = \multcut(\Pi_1,\dots,\Pi_n,\multcut(\mu(\Xi',\Pi_S),\Pi_S^t))$$
  which reduces to
  \begin{equation}\tag{$\star$}\label{eq:red-unfold}
    \inferrule*[right=\multc]
    {\inferrule*[right=\multc,leftskip=7em]
      {\deriv{\Pi_1}{\Delta_1\vdash B_1}\\\ldots
        \\ \deriv{\Pi_n}{\Delta_n\vdash B_n}\\
        \deriv{\mu(\Xi',\Pi_S)}{B_1,\ldots,B_n,\Gamma\vdash B\ S\ \vec t}}
      {\Delta_1,\ldots,\Delta_n,\Gamma\vdash B\ S\ \vec t}\\
      {\raisebox{.5em}{$\deriv{\Pi_S^{\vec t}}{B\ S\ \vec t\vdash S\ \vec t}$}}}
    {\Delta_1,\ldots,\Delta_n,\Gamma\vdash S\ \vec t}
  \end{equation}
  Now because $\multcut(\Pi_1,\ldots,\Pi_n,\Xi')$ is $\alpha$-reducible, it follows from our second induction hypothesis that
  $$\mu(\multcut(\Pi_1,\ldots,\Pi_n,\Xi'),\Pi_S) =\multcut(\Pi_1,\ldots,\Pi_n,\mu(\Xi',\Pi_S))$$
  is reducible, so it follows from assumption~(\ref{unfold-red:ass:2}) that the multicut above (\ref{eq:red-unfold}) is reducible. It then follows that $\mu(\Xi,\Pi_S)$ is reducible since its only reduction is reducible.
\end{proof}

\section{Cut reductions}
\label{appendix:cut-red}

If $n=0$, $\Xi$ reduces to the premise derivation $\Pi$.

For $n > 0$ we specify the reduction relation by case analysis on the last rule of $\Pi$:
\begin{itemize}
\item If $\Pi$ ends with a non-structural left rule (ie. not \cL or \wL) introducing the cut formula $B_i$ , we have the following subcases on the last rule of $\Pi_i$:
\begin{itemize}
\item {\em Essential case:} This is when $\Pi_i$ ends with a right rule.
\item {\em Inductive case:} This is when $\Pi$ ends with \IL.
\item {\em Left-commutative case:} This is when $\Pi_i$ ends with a left rule.
\item {\em Left-axiom case:} This is when $\Pi_i$ ends with an axiom rule.
\item {\em Left-multicut case:} This is when $\Pi_i$ ends with a multicut.
\end{itemize}

\item {\em Right-commutative case:} This is the case when 
$\Pi$ ends with a rule acting on a formula other than a cut formula.

\item {\em Structural case:} This is the case when $\Pi$ ends with 
a structural rule (\cL or \wL) on a cut formula. 

\item {\em Right-axiom case:} This is when $\Pi$ ends with an axiom rule.
\item {\em Right-multicut case:} This is when $\Pi$ ends with a multicut.
\end{itemize}

\subsection{\em \underline{Essential cases:}}
\emph{This is when $\Pi$ ends with a principal left rule and $\Pi_1$ ends with a (principal) right rule.}

\case[red-1-1]{\andR/\andL{i}}
If $\Pi_1$ and $\Pi$ are
 \begin{mathpar}
    {\inferrule*[right=\andR]
      {\deduce{\Delta_1\vdash B_1^1}
        {\Pi_1^1}
        \\ \deduce{\Delta_1\vdash B_1^2}
           {\Pi_1^2}}
      {\Delta_1\vdash B_1^1 \wedge B_1^2}}
    \and
    {\inferrule*[right=\andL{i}]
      {\deduce{B_1^i,B_2,\ldots,B_n,\Gamma\vdash C}
        {\Pi'}}{B_1^1 \wedge B_1^2,B_2,\ldots,B_n,\Gamma\vdash C}}
  \end{mathpar}
 then $\Xi$ reduces to $mc(\Pi_1^i,\Pi_2,\dots,\Pi_n,\Pi')$.

 \case[red-1-2]{\orR{i}/\orL}
If $\Pi_1$ and $\Pi$ are

\begin{mathpar}
  {\inferrule*[right=\orR{i}]
    {\deduce{\Delta_1\vdash B_1^i}{\Pi_1'}}
    {\Delta_1\vdash B_1^1 \vee B_1^2}}
\and
{\inferrule*[right=\orL]
  {\deduce{B_1^1,B_2,\ldots,B_n,\Gamma\vdash C}
    {\Pi^1}
    \\ \deduce{B_1^2,B_2,\ldots,B_n,\Gamma\vdash C}
       {\Pi^2}}
  {B_1^1 \vee B_1^2,B_2,\ldots,B_n,\Gamma\vdash C}}
\end{mathpar}

then $\Xi$ reduces to $mc(\Pi_1',\Pi_2,\dots,\Pi_n,\Pi^i)$.

\case[red-1-3]{\impR/\impL}
Suppose $\Pi_1$ and $\Pi$ are
\begin{mathpar}
\inferrule*[right=\impR]
        {\deduce{\Seq{B_1',\Delta_1}{B_1''}}
                {\Pi_1'}}{\Seq{\Delta_1}{B_1' \imp B_1''}}
\and
\inferrule*[right=\impL]
        {\deduce{\Seq{B_2,\ldots,B_n,\Gamma}{B_1'}}
                {\Pi'}
        \\ \deduce{\Seq{B_1'',B_2,\ldots,B_n,\Gamma}{C}}
                {\Pi''}}{\Seq{B_1' \imp B_1'',B_2,\ldots,B_n,\Gamma}{C}}
\end{mathpar}
Let $\Xi_1 = \multcut(\Pi_2,\dots,\Pi_n,\Pi')$, then $\Xi$ reduces to
\begin{mathpar}
    {\inferrule*[right=\multc]
      {{\raisebox{-1em}{${\inferrule*[right=\multc,leftskip=4em]
          {\deriv{\Xi_1}{\Delta_2,\ldots,\Delta_n,\Gamma\vdash B_1'}\\
            \deriv{\Pi_1'}{B_1',\Delta_1\vdash B_1''}}
          {\Delta_1,\ldots,\Delta_n,\Gamma\vdash B_1''}}$}}
        \quad
        \left\{{\deriv{\Pi_i}{\Seq{\Delta_i}{B_i}}}\right\}_{i \in \{2..n\}}
        \quad
        {\deriv{\Pi''}{\Seq{B_1'',B_2,\ldots,B_n,\Gamma}{C}}}}
      {\inferrule*[right=\cL *]
        {\Seq{\Delta_1,\ldots,\Delta_n,\Gamma,
          \Delta_2,\ldots,\Delta_n,\Gamma}{C}}
        {\Seq{\Delta_1,\ldots,\Delta_n,\Gamma}{C}}}}
\end{mathpar}

\case[red-1-4]{\allR/\allL}
Suppose $\Pi_1$ and $\Pi$ are
\begin{mathpar}
{\inferrule*[right=\allR]
  {\left\{\deriv{\Pi_1^s}{\Seq{\Delta_1}{B_1'[s/x]}}\right\}_{s \in \ground(\tau)}}
  {\Seq{\Delta_1}{\forall_\tau x.B_1'}}}\and
{\inferrule*[right=\allL]
  {\deduce{\Seq{B_1'[t/x],B_2,\ldots,B_n,\Gamma}{C}}{\Pi'}}
  {\Seq{\forall_\tau x.B_1',B_2,\ldots,B_n,\Gamma}{C}}}
\end{mathpar}
The derivation $\Xi$ reduces to
$mc(\Pi_1^t, \Pi_2,\dots,\Pi_n,\Pi').$

\case[red-1-5]{\exR/\exL}
Suppose $\Pi_1$ and $\Pi$ are
\begin{mathpar}
{\inferrule*[right=\exR]
  {{\deriv{\Pi_1'}{\Seq{\Delta_1}{B_1'[t/x]}}}}
  {\Seq{\Delta_1}{\exists_\tau x.B_1'}}
}\and
{\inferrule*[right=\exL]
  {\left\{\deriv{\Pi^s}
    {\Seq{B_1'[s/x],B_2,\ldots,B_n,\Gamma}{C}}\right\}_{s \in \ground(\tau)}
  }
  {\Seq{\exists_\tau x.B_1',B_2,\ldots, B_n,\Gamma}{C}}}
\end{mathpar}
Then $\Xi$ reduces to $mc(\Pi_1',\Pi_2,\dots,\Pi_n,\Pi^t).$

\case[red-1-6]{\DR/\DL}
Suppose $\Pi_1$ and $\Pi$ are
\begin{mathpar}
  {\inferrule*[right=\DR]
    {{\deriv{\Pi_1'}{\Seq{\Delta_1}{B'}}}}
    {\Seq{\Delta_1}{A}}}\quad\raisebox{1em}{$\dfn(H\defn B,A,\epsilon_\emptyset,B')$}\and
  {\inferrule*[right=\DL]
    {\left\{{\deriv{\Pi^{H\defn B}_{B'}}{\Seq{B',B_2,\ldots,B_n,\Gamma}{C}}}\st
      \dfn(H\defn B,A,\epsilon_\emptyset,B')\right\}}
    {\Seq{A,B_2,\ldots,B_n,\Gamma}{C}}}
\end{mathpar}
then $\Xi$ reduces to $\multcut(\Pi_1',\Pi_2,\ldots,\Pi_n,\Pi')$ where $\Pi'$ is $\Pi^{H\defn B}_{B'}$ for any clause $H\defn B$ such that $\dfn(H\defn B,A,\epsilon_\emptyset,B')$.


\case[red-1-7]{\nabR/\nabL}
Given $n,m\not\in\supp(B_1')$, if $\Pi_1$ and $\Pi$ are
\begin{mathpar}
  {\inferrule*[right=\nabR]
    {{\deriv{\Pi_1'}{\Delta_1\vdash B_1'[n/x]}}}
    {\Delta_1\vdash \nabla_\tau x.B_1'}}\and
  {\inferrule*[right=\nabL]
    {{\deriv{\Pi'}{B_1'[m/x],B_2,\ldots,B_n,\Gamma\vdash C}}}
    {\nabla_\tau x.B_1',B_2,\ldots,B_n,\Gamma\vdash C}}
\end{mathpar}
then by the renaming lemma (Lemma~\ref{renaming-lemma}), there exists a derivation $\Pi_1'' \sim\Pi_1'$ of $\Delta_1\vdash B_1'[m/x]$, and $\Xi$ reduces to $\multcut(\Pi_1'',\Pi_2,\ldots,\Pi_n,\Pi')$.

\subsection{\em \underline{Inductive case}}
\case[red-1-8]{$-$/\IL ({\em Inductive case})}
Suppose $\Pi_1$ is any derivation with end sequent $\Delta_1\vdash p\ \vec t$ and $\Pi$ is
\begin{mathpar}
  {\inferrule*[right=\IL]
    {\left\{\deriv{\Pi^{\vec t}_S}{B\ S\ \vec t\vdash S\ \vec t}\right\}_{\ground(t)}\\
      \deriv{\Pi'}{S\ \vec t,B_2,\ldots,B_n,\Gamma\vdash C}}
    {\Seq{p\ \vec t,B_2,\ldots,B_n,\Gamma}{C}}}
\end{mathpar}
then $\Xi$ reduces to 

\begin{mathpar}
  {\inferrule*[right=\multc]
    {{\deriv{\mu(\Pi_1,\Pi_S)}{\Delta_1\vdash S\ \vec t}}\\
      \left\{\deriv{\Pi_i}{\Delta_i\vdash B_i}\right\}_{i\in\{2\ldots n\}}\\
      {\deriv{\Pi'}{S\ \vec t,B_2,\ldots,B_n,\Gamma\vdash C}}}
    {\Delta_1,\ldots,\Delta_n, \Gamma\vdash C}}
\end{mathpar}


\subsection{\em \underline{Left-commutative cases:}}
\emph{This is when $\Pi$ is a principal left rule and $\Pi_1$ ends with a left rule.}

\case[red-2-1]{$\bulletL/\circL$}
Suppose $\Pi$ ends with a left rule other than \cL or \wL
acting on $B_1$ and $\Pi_1$ is
\begin{mathpar}
\inferrule*[right=$\bulletL$]
{\left\{
  \deriv{\Pi_1^i}{\Seq{\Delta_1^i}{B_1}}\right\}
}{\Seq{\Delta_1}{B_1}}
\end{mathpar}
where $\bulletL$ is any left rule except \impL and \IL.
Let $\Xi^i = mc(\Pi^i_1,\Pi_2,\dots,\Pi_n,\Pi)$.

Then $\Xi$ reduces to

\begin{mathpar}
\inferrule*[right=$\bulletL$]
{\left\{\deriv{\Xi^i}{\Seq{\Delta_1^i,\Delta_2,\ldots,\Delta_n,\Gamma}{C}}\right\}
}{\Seq{\Delta_1,\Delta_2,\ldots,\Delta_n,\Gamma}{C}}
\end{mathpar}

\case[red-2-2]{\impL$/\circL$}
Suppose $\Pi$ ends with a left rule other than \cL or \wL acting
on $B_1$ and $\Pi_1$ is
\begin{mathpar}
\inferrule*[right=\impL]
{\deduce{\Seq{\Delta_1'}{D_1'}}
                {\Pi_1'}
        \\ \deduce{\Seq{D_1'',\Delta_1'}{B_1}}
                {\Pi_1''}}{\Seq{D_1' \imp D_1'',\Delta_1'}{B_1}}
\end{mathpar}
Let $\Xi_1$ be  $mc(\Pi_1'',\Pi_2,\dots,\Pi_n,\Pi).$
Then $\Xi$ reduces to
\begin{mathpar}
\inferrule*[right=\impL]
{
 \infer=[$\wL$]
  {\Seq{\Delta_1',\Delta_2,\ldots,\Delta_n,\Gamma}{D_1'}}
  {\deduce{\Seq{\Delta_1'}{D_1'}}{\Pi_1'}}
  \\
\deduce{\Seq{D_1'',\Delta_1',\Delta_2,\ldots,\Delta_n,\Gamma}{C}}
                {\Xi_1}
}{\Seq{D_1' \imp D_1'',\Delta_1',\Delta_2,\ldots,\Delta_n,\Gamma}{C}}
\end{mathpar}

\case[red-2-3]{\IL$/\circL$}
Suppose $\Pi$ ends with a left rule other than \cL or \wL acting
on $B_1$ and $\Pi_1$ is
\begin{mathpar}
  {\inferrule*[right=\IL]
    {\left\{{\deriv{\Pi_S^t}{\Seq{B\ S\ \vec t}{S\ \vec t}}}\right\}_{\ground(t)}
      \\ \deriv{\Pi_1'}{\Seq{S\ \vec t,\Delta_1'}{B_1}}}
    {\Seq{p\ \vec t,\Delta_1'}{B_1}}}
\end{mathpar}
Let $\Xi_1$ be  $mc(\Pi_1',\Pi_2,\dots,\Pi_n,\Pi).$
Then $\Xi$ reduces to
\begin{mathpar}
  {\inferrule*[right=\IL]
    {\left\{{\deriv{\Pi_S^t}{\Seq{B\ S\ \vec t}{S\ \vec t}}}\right\}_{\ground(t)}
      \\
      \deriv{\Xi_1}{\Seq{S\ \vec t,\Delta_1',\Delta_2,\ldots,\Delta_n,\Gamma}{C}}}
    {\Seq{p\ \vec t,\Delta_1',\Delta_2,\ldots,\Delta_n,\Gamma}{C}}}
\end{mathpar}


\subsection{\em \underline{Right-commutative cases:}}
\emph{This is the case when $\Pi$ ends with a non-principal rule.}

\case[red-3-1]{$-/\circL$}
Suppose $\Pi$ is
\begin{mathpar}
\inferrule*[right=$\circL$]
{\left\{{\deriv{\Pi^i}{B_1,\ldots,B_n,\Gamma^i\vdash C}}\right\}}
{B_1,\ldots,B_n,\Gamma\vdash C}
\end{mathpar}
where $\circL$ is any left rule other than \impL and \IL (but including \cL and \wL) acting on a formula other than $B_1, \ldots, B_n$. 
Let $\Xi^i = mc(\Pi_1,\dots,\Pi_n,\Pi^i)$. 

The derivation $\Xi$ reduces to
\begin{mathpar}
  {\inferrule*[right=$\circL$]
    {\left\{
      {\deriv{\Xi^i}{\Seq{\Delta_1,\ldots,\Delta_n,\Gamma^i}{C}}}
      \right\}}
    {\Seq{\Delta_1,\ldots,\Delta_n,\Gamma}{C}}
  }
\end{mathpar}

\case[red-3-2]{$-/$\impL}
Suppose $\Pi$ is
\begin{mathpar}
\inferrule*[right=\impL]
        {\deduce{\Seq{B_1,\ldots,B_n,\Gamma'}{D_1}}
                {\Pi'}
        \\ \deduce{\Seq{B_1,\ldots,B_n,D_2,\Gamma'}{C}}
                {\Pi''}}{\Seq{B_1,\ldots,B_n, D_1 \imp D_2,\Gamma'}{C}}
\end{mathpar}
Let $\Xi_1 = mc(\Pi_1,\dots,\Pi_n,\Pi')$
and let $\Xi_2 = mc(\Pi_1,\dots,\Pi_n,\Pi'').$

Then $\Xi$ reduces to
\begin{mathpar}
\inferrule*[right=\impL]
        {\deduce{\Seq{\Delta_1,\ldots,\Delta_n,\Gamma'}{D_1}}
                {\Xi_1}
        \\ \deduce{\Seq{\Delta_1,\ldots,\Delta_n,D_2,\Gamma'}{C}}
                {\Xi_2}}{\Seq{\Delta_1,\ldots,\Delta_n,D_1 \imp D_2,\Gamma'}{C}}
\end{mathpar}

\case[red-3-3]{$-/$\IL}
Suppose $\Pi$ is
\begin{mathpar}
{\inferrule*[right=\IL]
  {\left\{{\deriv{\Pi_S^t}{\Seq{B\ S\ \vec t}{S\ \vec t}}}\right\}_{\ground(t)}
    \\
    \deriv{\Pi'}{\Seq{B_1,B_2,\ldots,B_n,S\ \vec t,\Gamma'}{C}}
  }
  {\Seq{B_1,\ldots,B_n, p\ \vec t,\Gamma'}{C}}
  }
\end{mathpar}
Let $\Xi' = mc(\Pi_1,\dots,\Pi_n,\Pi')$, then $\Xi$ reduces to
\begin{mathpar}
{\inferrule*[right=\IL]
  {\left\{{\deriv{\Pi_S^t}{\Seq{B\ S\ \vec t}{S\ \vec t}}}\right\}_{\ground(t)}
    \\
    \deriv{\Xi'}{\Seq{\Delta_1,\Delta_2,\ldots,\Delta_n,S\ \vec t,\Gamma'}{C}}
  }
  {
    {\Seq{\Delta_1,\ldots,\Delta_n,p\ \vec t,\Gamma'}{C}}
  }
}
\end{mathpar}

\case[red-3-4]{$-/$\impR}
Suppose $\Pi$ is
\begin{mathpar}
\inferrule*[right=\impR]
{{\deriv{\Pi'}{\Seq{B_1,\ldots,B_n,\Gamma,C_1}{C_2}}}}
{\Seq{B_1,\ldots,B_n,\Gamma}{C_1\imp C_2}}
\end{mathpar}
Let $\Xi' = mc(\Pi_1,\dots,\Pi_n,\Pi'),$
then $\Xi$ reduces to
\begin{mathpar}
\inferrule*[right=\impR]
{{\deriv{\Xi'}{\Seq{\Delta_1,\ldots,\Delta_n,\Gamma,C_1}{C_2}}}}
{\Seq{\Delta_1,\ldots,\Delta_n,\Gamma}{C_1\imp C_2}}
\end{mathpar}

\case[red-3-5]{$-/\circR$}
Suppose $\Pi$ is
\begin{mathpar}
\inferrule*[right=$\circR$]
{\left\{{\deriv{\Pi^i}{\Seq{B_1,\ldots,B_n,\Gamma^i}{C^i}}}\right\}}
{\Seq{B_1,\ldots,B_n,\Gamma}{C}}
\end{mathpar}
where $\circR$ is any right rule except for \impR. 
Let $\Xi^i = mc(\Pi_1,\dots,\Pi_n,\Pi^i),$
then $\Xi$ reduces to
\begin{mathpar}
\inferrule*[right=$\circR$]
{\left\{
\deriv{\Xi^i}{\Seq{\Delta_1,\ldots,\Delta_n,\Gamma^i}{C^i}}
\right\}}
{\Seq{\Delta_1,\ldots,\Delta_n,\Gamma}{C}}
\end{mathpar}


\subsection{\em \underline{Multicut cases:}}
\emph{This is the case when either $\Pi$ or $\Pi_1$ ends with a multicut.}

\case[red-4-1]{$\multcut/\circL$ (\emph{Left-multicut case})}
If $\Pi$ ends with a left rule other than \cL or \wL acting on $B_1$, and
$\Pi_1$ ends with a multicut and reduces to $\Pi_1'$,
then $\Xi$ reduces to
$\multcut(\Pi_1',\Pi_2,\dots,\Pi_n,\Pi).$

\case[red-4-2]{$-/\multcut$ (\emph{Right-multicut case})}

Suppose $\Pi$ is
\begin{mathpar}
  {\inferrule*[right=\multc]
    {\left\{{\deriv{\Pi^j}{\Seq{\{B_i\}_{i \in \mc I^j},\Gamma^j}{D^j}}}\right\}_{j \in [m]}
      \\ {\deriv{\Pi'}{\Seq{D^1,\ldots,D^m,\{B_i\}_{i \in \mc I'},\Gamma'}{C}}
        }}
    {\Seq{B_1,\ldots,B_n,\Gamma^1,\ldots,\Gamma^m,\Gamma'}{C}}}
\end{mathpar}
where $\mc I^1,\ldots,\mc I^m,\mc I'$ partition the formulas $B_1,\ldots, B_n$ among the premise derivations $\Pi^1, \ldots, \Pi^m$,$\Pi'$.
For $1 \leq j \leq m$ let $\Xi^j$ be
\begin{mathpar}
{\inferrule*[right=\multc]
  {\left\{{\deriv{\Pi_i}{\Seq{\Delta_i}{B_i}}}\right\}_{i \in \mc I^j}
    \\ {\deriv{\Pi^j}{\Seq{\{B_i\}_{i \in \mc I^j},\Gamma^j}{D^j}}}}
  {\Seq{\{\Delta_i\}_{i \in \mc I^j},\Gamma^j}{D^j}}}
\end{mathpar}
Then $\Xi$ reduces to

\begin{mathpar}
{\inferrule*[right=\multc]
  {\left\{{\deriv{\Xi^j}{\Seq{\{\Delta_i\}_{i\in\mc I^j},\Gamma^j}{D^j}}}\right\}_{j \in [m]}
    \\ \left\{{\deriv{\Pi_i}{\Seq{\Delta_i}{B_i}}
    }\right\}_{i \in \mc I'}
    \\ {\deriv{\Pi'}{\Seq{D^1,\ldots,D^m,\{B_i\}_{i\in\mc I'},\Gamma'}{C}}}}
  {\Seq{\Delta_1,\ldots,\Delta_n,\Gamma^1,\ldots\Gamma^m,\Gamma'}{C}}}
\end{mathpar}


\subsection{{\em \underline{Structural cases:}}}
\emph{This is the case when $\Pi$ is a principal structural rule.}

\case[red-5-1]{$-/$\cL}
Supose $\Pi$ is
\begin{mathpar}
  {\inferrule*[right=\cL]
    {\deduce{\Seq{B_1,B_1,B_2,\ldots,B_n,\Gamma}{C}}
      {\Pi'}}
    {\Seq{B_1,B_2,\ldots,B_n,\Gamma}{C}}
    }
\end{mathpar}

Let $\Xi' = mc(\Pi_1,\Pi_1,\Pi_2,\dots,\Pi_n,\Pi').$ 
Then $\Xi$ reduces to
\begin{mathpar}
  {\inferrule*[right=\cL]
    {\deduce{\Seq{\Delta_1,\Delta_1,\Delta_2,\ldots,\Delta_n,\Gamma}{C}}{\Xi'}}
    {\Seq{\Delta_1,\Delta_2,\ldots,\Delta_n,\Gamma}{C}}}
\end{mathpar}

\case[red-5-2]{$-/$\wL}
Suppose $\Pi$ is 
\begin{mathpar}
  {\inferrule*[right=\wL]    
    { \deduce{\Seq{B_2,\dots,B_n,\Gamma}{C}}{\Pi'}}
    {\Seq{B_1,B_2,\dots,B_n,\Gamma}{C}}}
\end{mathpar}

Let $\Xi' = mc(\Pi_2,\dots,\Pi_n,\Pi').$
Then $\Xi$ reduces to:

\begin{mathpar}
{\inferrule*[right=\wL]
  { \deduce{\Seq{\Delta_2,\ldots,\Delta_n,\Gamma}{C}}{\Xi'}}
  {\Seq{\Delta_1,\Delta_2,\ldots,\Delta_n,\Gamma}{C}}
  }
\end{mathpar}


\subsection{\em \underline{Axiom cases:}}
\emph{This is the case where either $\Pi$ or $\Pi_1$ ends with \axiom.}

\case[red-6-1]{\axiom$/\circL$ (\emph{Left-axiom case})}

Suppose $\Pi$ ends with a left rule on $B_1$ and 
$\Pi_1$ ends with $\axiom$. This means that $\Delta_1 = \{B_1'\}$ for some $B_1'\approx B_1$, thus $\Xi$ reduces to
\begin{mathpar}
  {\inferrule*[right=\multc]
    {\left\{\deriv{\Pi_i}{\Delta_i\vdash B_i}\right\}_{i\in[2,n]}\\
      \deriv{\Pi'}{\Gamma,B_1',\ldots,B_n\vdash C}}
    {\Gamma,B_1',\Delta_2,\ldots,\Delta_n\vdash C}}
\end{mathpar}
where $\Pi'\sim\Pi$ is the derivation obtained from the renaming lemma (Lemma~\ref{renaming-lemma}).

\case[red-6-2]{$-/$\axiom (\emph{Right-axiom case})}

If $\Pi$ ends with the $\axiom$ rule, then $n=1$,
$\Gamma = \emptyset$ and $C = B_1'$ for $B_1\approx B_1'$.
In this case, $\Xi$ reduces to $\Pi_1',$ where $\Pi_1'\sim\Pi_1$ is obtained from the renaming lemma (Lemma~\ref{renaming-lemma}).

\section{The reducibility lemma}
\label{appendix:reducibility-lemma}

\lemreducibilitylemma*
\begin{proof}
  We proceed first by induction on the \textlabel{index $\Ind(\Pi)$}{ind-index}, then on the \textlabel{height $\height(\Pi)$}{ind-height}, then by subordinate induction on the \textlabel{number of cut formulas $n$}{ind-cut-formulas}.

  \case{Case $n = 0$} In this case $\Xi$ reduces to $\Pi$, so it suffices to show that $\Pi$ is reducible. We use the following trick: for any derivation $\Psi$ with height less than $\height(\Pi)$, the second induction hypothesis applies to the multicut $\multcut(\Psi)$, and since this reduces to $\Psi$ it must be that $\Psi$ is reducible. We can now show $\Pi$ is reducible by case analysis on $\Pi$:

  \subcase{Subcase \impR} Suppose $\Pi$ is a derivation of $\Gamma\vdash A\imp B$ with premise derivation $\Pi'$, let $\alpha = \lvl(A)$, and let $\Upsilon$ be an $\alpha$-reducible derivation of $\Delta\vdash A$, then since $\height(\Pi')<\height(\Pi)$, $\Pi'$ is reducible, and furthermore the second induction hypothesis applies to $\multcut(\Upsilon,\Pi')$, so $\Pi$ is reducible.

  \subcase{Subcase \impL} If $\Pi$ ends with \impL with major premise derivation $\Pi'$ and minor premise derivation $\Psi$, then since $\height(\Pi'),\height(\Psi)<\height(\Pi)$, both $\Psi$ and $\Pi'$ are reducible, and by the normalization lemma (Lemma~\ref{normalization-lemma}) $\Psi$ must be normalizable, thus $\Pi$ is reducible.

  \subcase{Subcase \IL} If $\Pi$ ends with \IL with major premise derivation $\Pi'$ and minor premise derivations $\{\Pi_{\vec u}\}_{\ground(\vec u)}$, then since $\height(\Pi'),\height(\Pi_{\vec u})<\height(\Pi)$, all $\Pi_{\vec u}$ and $\Pi'$ are reducible, and by the normalization lemma (Lemma~\ref{normalization-lemma}) $\Pi_{\vec u}$ must be normalizable, thus $\Pi$ is reducible.

  \subcase{Subcase \multc} If $\Pi$ is a multicut $\multcut(\Psi_1,\ldots,\Psi_n,\Pi')$, then each $\Psi_i$ is reducible since $\height(\Psi_i)<\height(\Pi)$, so the second induction hypothesis applies.

  \subcase{Otherwise} If $\Pi$ has premise derivations $\{\Psi_i\}_{i\in\mc I}$, then since $\height(\Psi_i)<\height(\Pi)$ then $\Psi_i$ are reducible and so is $\Pi$.
  
  \case{Case $n > 0$}
  We proceed by subordinate induction on the \textlabel{$\alpha_i$-reducibility of $\Pi_i$}{ind-reducibility} for $i=1\ldots n$, and proceed by case analysis on the reduction of $\Xi$. In what follows, when $\Pi$ ends with a principal left rule acting on $B_i$, we assume for clarity and without loss of generality that $i = 1$.

  \subsection{\em \underline{Essential cases:}}

  \emph{Suppose $\Pi$ ends with a principal left rule and $\Pi_1$ ends with a (principal) right rule.}

  \refcase[red-1-1]{\andR/\andL}
  If $\Pi_1$ and $\Pi$ are
  \begin{mathpar}
    {\inferrule*[right=\andR]
      {\deduce{\Seq{\Delta_1}{B_1^1}}
        {\Pi_1^1}
        \\ \deduce{\Seq{\Delta_1}{B_1^2}}
           {\Pi_1^2}}
      {\Seq{\Delta_1}{B_1^1 \land B_1^2}}}    \and
    {\inferrule*[right=\andL{i}]
      {\deduce{\Seq{B_1^i,B_2,\ldots,B_n,\Gamma}{C}}
        {\Pi'}}{\Seq{B_1^1 \land B_1^2,B_2,\ldots,B_n,\Gamma}{C}}
    }
  \end{mathpar}
  then the reduction $mc(\Pi_1^i,\Pi_2,\dots,\Pi_n,\Pi')$ is reducible from the induction hypothesis \ref{ind-height} because $\Pi_1^i$ is reducible and $\height(\Pi')<\height(\Pi)$.

\refcase[red-1-2]{\orR{i}/\orL}
If $\Pi_1$ and $\Pi$ are
\begin{mathpar}
\inferrule*[right=\orR{i}]
        {\deduce{\Seq{\Delta_1}{B_1^i}}
                {\Pi_1'}}{\Seq{\Delta_1}{B_1^1 \lor B_1^2}}
\and
\inferrule*[right=\orL]
        {\deduce{\Seq{B_1^1,B_2,\ldots,B_n,\Gamma}{C}}
                {\Pi^1}
        \\ \deduce{\Seq{B_1^2,B_2,\ldots,B_n,\Gamma}{C}}
                {\Pi^2}}{\Seq{B_1^1 \lor B_1^2,B_2,\ldots,B_n,\Gamma}{C}}
\end{mathpar}
then the reduction $mc(\Pi_1',\Pi_2,\dots,\Pi_n,\Pi^i)$ is reducible from the induction hypothesis \ref{ind-height} because $\Pi_1'$ is reducible and $\height(\Pi^i)<\height(\Pi)$.

\refcase[red-1-3]{\impR/\impL}
Suppose $\Pi_1$ and $\Pi$ are
\begin{mathpar}
\inferrule*[right=\impR]
        {\deduce{\Seq{B_1',\Delta_1}{B_1''}}
                {\Pi_1'}}{\Seq{\Delta_1}{B_1' \imp B_1''}}
\and
\inferrule*[right=\impL]
        {\deduce{\Seq{B_2,\ldots,B_n,\Gamma}{B_1'}}
                {\Pi'}
        \\ \deduce{\Seq{B_1'',B_2,\ldots,B_n,\Gamma}{C}}
                {\Pi''}}{\Seq{B_1' \imp B_1'',B_2,\ldots,B_n,\Gamma}{C}}
\end{mathpar}
then to show that the reduction $\multcut(\Xi_1',\Pi_2,\ldots,\Pi_n,\Pi'')$ is reducible it sufficies to show by the induction hypothesis \ref{ind-height} that the multicut $\Xi_1' = \multcut(\Xi_1,\Pi_1')$ is reducible. However, this follows from the reducibility of $\Pi_1$, given that the multicut $\Xi_1 = \multcut(\Pi_2,\ldots,\Pi_n,\Pi')$ is reducible, which itself follows from the induction hypothesis \ref{ind-height}, where we note that $\height(\Pi'),\height(\Pi'') < \height(\Pi)$.

\refcase[red-1-4]{$\forallR/\forallL$}
Suppose $\Pi_1$ and $\Pi$ are
\begin{mathpar}
{\inferrule*[right=\allR]
  {\left\{\deriv{\Pi_1^s}{\Seq{\Delta_1}{B_1'[s/x]}}\right\}_{s \in \ground(\tau)}}
  {\Seq{\Delta_1}{\forall_\tau x.B_1'}}}\and
{\inferrule*[right=\allL]
  {\deduce{\Seq{B_1'[t/x],B_2,\ldots,B_n,\Gamma}{C}}{\Pi'}}
  {\Seq{\forall_\tau x.B_1',B_2,\ldots,B_n,\Gamma}{C}}}
\end{mathpar}
then the reduction
$mc(\Pi_1^t, \Pi_2,\dots,\Pi_n,\Pi')$
is reducible by the induction hypothesis \ref{ind-height}, since the reducibility of $\Pi_1^t$ follows from the reducibility of $\Pi_1$, and $\height(\Pi')<\height(\Pi)$.

\refcase[red-1-5]{$\existsR/\existsL$}
Suppose $\Pi_1$ and $\Pi$ are
\begin{mathpar}
{\inferrule*[right=\exR]
  {{\deriv{\Pi_1'}{\Seq{\Delta_1}{B_1'[t/x]}}}}
  {\Seq{\Delta_1}{\exists_\tau x.B_1'}}
}\and
{\inferrule*[right=\exL]
  {\left\{\deriv{\Pi^s}
    {\Seq{B_1'[s/x],B_2,\ldots,B_n,\Gamma}{C}}\right\}_{s \in \ground(\tau)}
  }
  {\Seq{\exists_\tau x.B_1',B_2,\ldots, B_n,\Gamma}{C}}}
\end{mathpar}
then the reduction $\multcut(\Pi_1',\Pi_2,\dots,\Pi_n,\Pi^t)$ is reducible by the induction hypothesis \ref{ind-height}, since the reducibility of $\Pi_1$ implies the reducibility of $\Pi_1'$, and $\height(\Pi^t)<\height(\Pi)$.

\refcase[red-1-6]{\DR/\DL}
Suppose $\Pi_1$ and $\Pi$ are
\begin{mathpar}
  {\inferrule*[right=\DR]
    {{\deriv{\Pi_1'}{\Seq{\Delta_1}{B'}}}}
    {\Seq{\Delta_1}{A}}}\quad\raisebox{1em}{$\dfn(H\defn B,A,\epsilon_\emptyset,B')$}\and
  {\inferrule*[right=\DL]
    {\left\{{\deriv{\Pi^{H\defn B}_{B'}}{\Seq{B',B_2,\ldots,B_n,\Gamma}{C}}}\st
      \dfn(H\defn B,A,\epsilon_\emptyset,B')\right\}}
    {\Seq{A,B_2,\ldots,B_n,\Gamma}{C}}}
\end{mathpar}
then the reduction $\multcut(\Pi_1',\Pi_2,\ldots,\Pi_n,\Pi')$ where $\Pi'$ is $\Pi^{H\defn B}_{B'}$ for any clause $H\defn B$ such that $\dfn(H\defn B,A,\epsilon_\emptyset,B')$ is reducible by the induction hypothesis \ref{ind-height} since the reducibility of $\Pi_1'$ is implied by the reducibility of $\Pi_1$, and $\height(\Pi')<\height(\Pi)$.



\refcase[red-1-7]{\nabR/\nabL}
Given $n,m\not\in\supp(B_1')$, if $\Pi_1$ and $\Pi$ are
\begin{mathpar}
  {\inferrule*[right=\nabR]
    {{\deriv{\Pi_1'}{\Delta_1\vdash B_1'[n/x]}}}
    {\Delta_1\vdash \nabla_\tau x.B_1'}}\and
  {\inferrule*[right=\nabL]
    {{\deriv{\Pi'}{B_1'[m/x],B_2,\ldots,B_n,\Gamma\vdash C}}}
    {\nabla_\tau x.B_1',B_2,\ldots,B_n,\Gamma\vdash C}}
\end{mathpar}
then by the renaming lemma (Lemma~\ref{renaming-lemma}), there exists a derivation $\Pi_1'' \sim\Pi_1'$ of $\Delta_1\vdash B_1'[m/x]$, and the reduction $\multcut(\Pi_1'',\Pi_2,\ldots,\Pi_n,\Pi')$ is reducible by the induction hypothesis \ref{ind-height}, since $\Ind(\Pi') = \Ind(\Pi)$ and $\height(\Pi')<\height(\Pi)$, and furthermore, by the renaming congruence lemma (Lemma~\ref{renaming-congruence-lemma}) $\Pi_1''$ is reducible.

\subsection{\em \underline{Inductive case}}
\refcase[red-1-8]{-/\IL ({\em Inductive case})}
Suppose $\Pi_1$ is any derivation with end sequent $\Delta_1\vdash p\ \vec t$ and $\Pi$ is
\begin{mathpar}
  {\inferrule*[right=\IL]
    {\left\{\deriv{\Pi^t_S}{B\ S\ \vec t\vdash S\ \vec t}\right\}_{\ground(t)}\\
      \deriv{\Pi'}{S\ \vec t,B_2,\ldots,B_n,\Gamma\vdash C}}
    {\Seq{p\ \vec t,B_2,\ldots,B_n,\Gamma}{C}}}
\end{mathpar}
then to show that the reduction $\multcut(\mu(\Pi_1,\Pi_S),\Pi_2,\ldots,\Pi_n,\Pi')$ is reducible it suffices to show by the induction hypothesis \ref{ind-height} that $\mu(\Pi_1,\Pi_S)$ is reducible. This follows from lemma~\ref{unfolded-reducibility-lemma} provided the following assumptions are satisfied:

\subcase{\em Assumption 1} Note that $\Ind(\Pi_S^t) < \Ind(\Pi)$, so from the induction hypothesis \ref{ind-index} it follows that the nullary cut $\multcut(\Pi_S^{\vec t})$ is reducible, which implies that $\Pi_S^{\vec t}$ is reducible, and therefore also normalizable by Lemma~\ref{normalization-lemma}.
\subcase{\em Assumption 2} Note that $\Ind(\Pi_S^t) < \Ind(\Pi)$, so from the induction hypothesis \ref{ind-index} it follows that for any reducible $\Upsilon$, $\multcut(\Upsilon,\Pi_S^t)$ is reducible.

\subcase{\em Assumption 3} Note that $\Ind(\init) = 0<\Ind(\Pi)$, so by inductive hypothesis \ref{ind-index} we have that for any reducible $\Upsilon$, $\multc(\Upsilon,\init)$ is reducible.

\subsection{\em \underline{Left-commutative cases:}}
\emph{Suppose $\Pi$ is a principal left rule and $\Pi_1$ ends with a left rule.}

\refcase[red-2-1]{$\bulletL/\circL$}
Suppose $\Pi$ ends with a left rule other than \cL or \wL
acting on $B_1$ and $\Pi_1$ is
\begin{mathpar}
\inferrule*[right=$\bulletL$]
{\left\{
  \deriv{\Pi_1^i}{\Seq{\Delta_1^i}{B_1}}\right\}
}{\Seq{\Delta_1}{B_1}}
\end{mathpar}
where $\bulletL$ is any left rule except \impL and \IL. In this case each $\Pi^i_1$ is $\alpha_1$-reducible (since every premise of $\Pi_1$ is a major premise), so by the induction hypothesis \ref{ind-reducibility}, each derivation $\Xi^i = mc(\Pi^i_1,\Pi_2,\dots,\Pi_n,\Pi)$ is reducible. It follows that the reduction
\begin{mathpar}
\inferrule*[right=$\bulletL$]
{\left\{\deriv{\Xi^i}{\Seq{\Delta_1^i,\Delta_2,\ldots,\Delta_n,\Gamma}{C}}\right\}
}{\Seq{\Delta_1,\Delta_2,\ldots,\Delta_n,\Gamma}{C}}
\end{mathpar}
is reducible.

\refcase[red-2-2]{\impL$/\circL$}
Suppose $\Pi$ ends with a left rule other than \cL or \wL acting
on $B_1$ and $\Pi_1$ is
\begin{mathpar}
\inferrule*[right=\impL]
{\deduce{\Seq{\Delta_1'}{D_1}}
                {\Pi_1'}
        \\ \deduce{\Seq{D_2,\Delta_1'}{B_1}}
                {\Pi_1''}}{\Seq{D_1 \imp D_2,\Delta_1'}{B_1}}
\end{mathpar}
In this case $\Pi_1'$ is normalizable and $\Pi_1''$ is $\alpha_1$-reducible, so by the induction hypothesis \ref{ind-reducibility}, it follows that the derivation $\Xi_1 = mc(\Pi_1'',\Pi_2,\dots,\Pi_n,\Pi)$ is reducible. Furthermore, $\weakL(\Pi_1')$ is normalizable, thus the reduction
\begin{mathpar}
\inferrule*[right=\impL]
{
 \infer=[\weakL]
  {\Seq{\Delta_1',\Delta_2,\ldots,\Delta_n,\Gamma}{D_1}}
  {\deduce{\Seq{\Delta_1'}{D_1}}{\Pi_1'}}
  \\
\deduce{\Seq{D_2,\Delta_1',\Delta_2,\ldots,\Delta_n,\Gamma}{C}}
                {\Xi_1}
}{\Seq{D_1 \imp D_2,\Delta_1',\Delta_2,\ldots,\Delta_n,\Gamma}{C}}
\end{mathpar}
is reducible.

\refcase[red-2-3]{\IL/$\circL$}
Suppose $\Pi$ ends with a left rule $\circL$ other than \cL or \wL acting
on $B_1$ and $\Pi_1$ is
\begin{mathpar}
  {\inferrule*[right=\IL]
    {\left\{{\deriv{\Pi_S^t}{\Seq{B\ S\ \vec t}{S\ \vec t}}}\right\}_{\ground(t)}
      \\ \deriv{\Pi_1'}{\Seq{S\ \vec t,\Delta_1'}{B_1}}}
    {\Seq{p\ \vec t,\Delta_1'}{B_1}}}
\end{mathpar}
In this case the $\Pi_S^t$ are normalizable and $\Pi_1'$ is $\alpha_1$-reducible, so by the induction hypothesis \ref{ind-reducibility}, it follows that the derivation $\Xi_1 = mc(\Pi_1',\Pi_2,\dots,\Pi_n,\Pi)$ is reducible. Thus the reduction

\begin{mathpar}
  {\inferrule*[right=\IL]
    {\left\{{\deriv{\Pi_S^t}{\Seq{B\ S\ \vec t}{S\ \vec t}}}\right\}_{\ground(t)}
      \\
      \deriv{\Xi_1}{\Seq{S\ \vec t,\Delta_1',\Delta_2,\ldots,\Delta_n,\Gamma}{C}}}
    {\Seq{p\ \vec t,\Delta_1',\Delta_2,\ldots,\Delta_n,\Gamma}{C}}}
\end{mathpar}
is reducible.

\subsection{\em \underline{Right commutative cases:}}
\emph{Suppose $\Pi$ ends with a non-principal rule.}

\refcase[red-3-1]{$-/\circL$}
Suppose $\Pi$ is
\begin{mathpar}
\inferrule*[right=$\circL$]
{\left\{{\deriv{\Pi^i}{B_1,\ldots,B_n,\Gamma^i\vdash C}}\right\}}
{B_1,\ldots,B_n,\Gamma\vdash C}
\end{mathpar}
where $\circL$ is any left rule other than \impL and \IL (but including \cL and \wL) acting on a formula other than $B_1, \ldots, B_n$. Then since $\height(\Pi^i)<\height(\Pi)$, it follows from the induction hypothesis \ref{ind-height} that each $\Xi^i = mc(\Pi_1,\dots,\Pi_n,\Pi^i)$ is reducible, and therefore the reduction
\begin{mathpar}
  {\inferrule*[right=$\circL$]
    {\left\{
      {\deriv{\Xi^i}{\Seq{\Delta_1,\ldots,\Delta_n,\Gamma^i}{C}}}
      \right\}}
    {\Seq{\Delta_1,\ldots,\Delta_n,\Gamma}{C}}
  }
\end{mathpar}
is reducible.

\refcase[red-3-2]{$-/$\impL}
Suppose $\Pi$ is
\begin{mathpar}
\inferrule*[right=\impL]
        {\deduce{\Seq{B_1,\ldots,B_n,\Gamma'}{D_1}}
                {\Pi'}
        \\ \deduce{\Seq{B_1,\ldots,B_n,D_2,\Gamma'}{C}}
                {\Pi''}}{\Seq{B_1,\ldots,B_n, D_1 \imp D_2,\Gamma'}{C}}
\end{mathpar}
Since $\height(\Pi'),\height(\Pi'')<\height(\Pi)$ it follows that $\Xi_1 = mc(\Pi_1,\dots,\Pi_n,\Pi')$ and $\Xi_2 = mc(\Pi_1,\dots,\Pi_n,\Pi'')$ are both reducible by the inductive hypothesis \ref{ind-height}, and thus $\Xi_1$ is normalizable by Lemma~\ref{normalization-lemma} and the reduction
\begin{mathpar}
  \inferrule*[right=\impL]
  {\deduce{\Seq{\Delta_1,\ldots,\Delta_n,\Gamma'}{D_1}}
    {\Xi_1}
    \\ \deduce{\Seq{\Delta_1,\ldots,\Delta_n,D_2,\Gamma'}{C}}
    {\Xi_2}}{\Seq{\Delta_1,\ldots,\Delta_n,D_1 \imp D_2,\Gamma'}{C}}
\end{mathpar}
is reducible.

\refcase[red-3-3]{$-/$\IL}
Suppose $\Pi$ is
\begin{mathpar}
{\inferrule*[right=\IL]
  {\left\{{\deriv{\Pi_S^t}{\Seq{B\ S\ \vec t}{S\ \vec t}}}\right\}_{\ground(t)}
    \\
    \deriv{\Pi'}{\Seq{B_1,B_2,\ldots,B_n,S\ \vec t,\Gamma'}{C}}
  }
  {\Seq{B_1,\ldots,B_n, p\ \vec t,\Gamma'}{C}}
  }
\end{mathpar}
then since $\height(\Pi')<\height(\Pi)$ it follows from the induction hypothesis \ref{ind-height} that $\Xi' = mc(\Pi_1,\dots,\Pi_n,\Pi')$ is reducible and therefore the reduction
\begin{mathpar}
{\inferrule*[right=\IL]
  {\left\{{\deriv{\Pi_S^t}{\Seq{B\ S\ \vec t}{S\ \vec t}}}\right\}_{\ground(t)}
    \\
    \deriv{\Xi'}{\Seq{\Delta_1,\Delta_2,\ldots,\Delta_n,S\ \vec t,\Gamma'}{C}}
  }
  {
    {\Seq{\Delta_1,\ldots,\Delta_n,p\ \vec t,\Gamma'}{C}}
  }
}
\end{mathpar}
is reducible.

\refcase[red-3-4]{$-/$\impR}
Suppose $\Pi$ is
\begin{mathpar}
\inferrule*[right=\impR]
{{\deriv{\Pi'}{\Seq{B_1,\ldots,B_n,\Gamma,C_1}{C_2}}}}
{\Seq{B_1,\ldots,B_n,\Gamma}{C_1\imp C_2}}
\end{mathpar}
Since $\height(\Pi')<\height(\Pi)$ it follows that $\Xi' = mc(\Pi_1,\dots,\Pi_n,\Pi')$ is reducible by the induction hypothesis \ref{ind-height}.
Furthermore, for any reducible derivation $\Upsilon$ of a sequent $\Seq{\Delta}{C_1}$,
since $mc(\Upsilon,\Xi')$ reduces to $mc(\Pi_1,\ldots,\Pi_n,\Upsilon,\Pi')$,
it again follows from induction hypothesis \ref{ind-height} that $mc(\Upsilon,\Xi')$ is reducible, and therefore the reduction reduction
\begin{mathpar}
\inferrule*[right=\impR]
{{\deriv{\Xi'}{\Seq{\Delta_1,\ldots,\Delta_n,\Gamma,C_1}{C_2}}}}
{\Seq{\Delta_1,\ldots,\Delta_n,\Gamma}{C_1\imp C_2}}
\end{mathpar}
is reducible.

\refcase[red-3-5]{$-/\circR$}
Suppose $\Pi$ is
\begin{mathpar}
\inferrule*[right=$\circR$]
{\left\{{\deriv{\Pi^i}{\Seq{B_1,\ldots,B_n,\Gamma^i}{C^i}}}\right\}}
{\Seq{B_1,\ldots,B_n,\Gamma}{C}}
\end{mathpar}
where $\circR$ is any right rule. 
Since $\height(\Pi^i)<\height(\Pi)$ it follows that $\Xi^i = mc(\Pi_1,\dots,\Pi_n,\Pi^i)$ is reducible by the induction hypothesis \ref{ind-height}, and therefore the reduction
\begin{mathpar}
\inferrule*[right=$\circR$]
{\left\{
\deriv{\Xi^i}{\Seq{\Delta_1,\ldots,\Delta_n,\Gamma^i}{C^i}}
\right\}}
{\Seq{\Delta_1,\ldots,\Delta_n,\Gamma}{C}}
\end{mathpar}
is reducible.

\subsection{\em \underline{Multicut cases:}}
\emph{Suppose either $\Pi$ or $\Pi_1$ ends with a multicut.}

\refcase[red-4-1]{$\multcut/\circL$ (\emph{Left-multicut case})}
If $\Pi$ ends with a left rule other than \cL or \wL acting on $B_1$, and
$\Pi_1$ ends with a multicut that reduces to $\Pi_1'$ (which is $\alpha_1$-reducible),
then the reduction
$\multcut(\Pi_1',\Pi_2,\dots,\Pi_n,\Pi).$
is reducible by the induction hypothesis \ref{ind-reducibility}.

\refcase[red-4-2]{$-/\multcut$ (\emph{Right-multicut case})}

Suppose $\Pi$ is
\begin{mathpar}
  {\inferrule*[right=\multc]
    {\left\{{\deriv{\Pi^j}{\Seq{\{B_i\}_{i \in \mc I^j},\Gamma^j}{D^j}}}\right\}_{j \in [m]}
      \\ {\deriv{\Pi'}{\Seq{D^1,\ldots,D^m,\{B_i\}_{i \in \mc I'},\Gamma'}{C}}
        }}
    {\Seq{B_1,\ldots,B_n,\Gamma^1,\ldots,\Gamma^m,\Gamma'}{C}}}
\end{mathpar}
where $\mc I^1,\ldots,\mc I^m,\mc I'$ partition the formulas $B_1,\ldots, B_n$ among the premise derivations $\Pi^1, \ldots, \Pi^m$,$\Pi'$.
For $1 \leq j \leq m$, note that $\height(\Pi^j)<\height(\Pi)$ and thus each derivation $\Xi^j$ below
\begin{mathpar}
{\inferrule*[right=\multc]
  {\left\{{\deriv{\Pi_i}{\Seq{\Delta_i}{B_i}}}\right\}_{i \in \mc I^j}
    \\ {\deriv{\Pi^j}{\Seq{\{B_i\}_{i \in \mc I^j},\Gamma^j}{D^j}}}}
  {\Seq{\{\Delta_i\}_{i \in \mc I^j},\Gamma^j}{D^j}}}
\end{mathpar}
is reducible by the induction hypothesis \ref{ind-height}.

Furthermore, since $\height(\Pi')<\height(\Pi)$, the reduction
\begin{mathpar}
{\inferrule*[right=\multc]
  {\left\{{\deriv{\Xi^j}{\Seq{\{\Delta_i\}_{i\in\mc I^j},\Gamma^j}{D^j}}}\right\}_{j \in [m]}
    \\ \left\{{\deriv{\Pi_i}{\Seq{\Delta_i}{B_i}}
    }\right\}_{i \in \mc I'}
    \\ {\deriv{\Pi'}{\Seq{D^1,\ldots,D^m,\{B_i\}_{i\in\mc I'},\Gamma'}{C}}}}
  {\Seq{\Delta_1,\ldots,\Delta_n,\Gamma^1,\ldots\Gamma^m,\Gamma'}{C}}}
\end{mathpar}
is reducible by the induction hypothesis \ref{ind-height}.

\subsection{\em \underline{Structural cases:}}
\emph{Suppose $\Pi$ is a principal structural rule.}

\refcase[red-5-1]{$-/$\cL}
Supose $\Pi$ is
\begin{mathpar}
  {\inferrule*[right=\cL]
    {\deduce{\Seq{B_1,B_1,B_2,\ldots,B_n,\Gamma}{C}}
      {\Pi'}}
    {\Seq{B_1,B_2,\ldots,B_n,\Gamma}{C}}
    }
\end{mathpar}
then since $\height(\Pi')<\height(\Pi)$ it follows from induction hypothesis \ref{ind-height} that $\Xi' = mc(\Pi_1,\Pi_1,\Pi_2,\dots,\Pi_n,\Pi')$ is reducible, and thus the reduction
\begin{mathpar}
  {\inferrule*[right=\cL]
    {\deduce{\Seq{\Delta_1,\Delta_1,\Delta_2,\ldots,\Delta_n,\Gamma}{C}}{\Xi'}}
    {\Seq{\Delta_1,\Delta_2,\ldots,\Delta_n,\Gamma}{C}}}
\end{mathpar}
is reducible.

\refcase[red-5-2]{$-/$\wL}
Suppose $\Pi$ is 
\begin{mathpar}
  {\inferrule*[right=\wL]    
    { \deduce{\Seq{B_2,\dots,B_n,\Gamma}{C}}{\Pi'}}
    {\Seq{B_1,B_2,\dots,B_n,\Gamma}{C}}}
\end{mathpar}

then since $\height(\Pi')<\height(\Pi)$ it follows from induction hypothesis \ref{ind-height} that $\Xi' = mc(\Pi_2,\dots,\Pi_n,\Pi')$ is reducible and thus the reduction
\begin{mathpar}
{\inferrule*[right=\wL]
  { \deduce{\Seq{\Delta_2,\ldots,\Delta_n,\Gamma}{C}}{\Xi'}}
  {\Seq{\Delta_1,\Delta_2,\ldots,\Delta_n,\Gamma}{C}}
  }
\end{mathpar}
is reducible.

\subsection{\em \underline{Axiom cases:}}
\emph{Suppose either $\Pi$ or $\Pi_1$ ends with \axiom.}
\refcase[red-6-1]{\axiom$/\circL$ (\emph{Left-axiom case})}

Suppose $\Pi$ ends with a left rule on $B_1$ and 
$\Pi_1$ ends with $\axiom$. This means that $\Delta_1 = \{B_1'\}$ for some $B_1'\approx B_1$, thus $\Xi$ reduces to
\begin{mathpar}
  {\inferrule*[right=\multc]
    {\left\{\deriv{\Pi_i}{\Delta_i\vdash B_i}\right\}_{i\in[2,n]}\\
      \deriv{\Pi'}{\Gamma,B_1',\ldots,B_n\vdash C}}
    {\Gamma,B_1',\Delta_2,\ldots,\Delta_n\vdash C}}
\end{mathpar}
where $\Pi'\sim\Pi$ is the derivation obtained from the renaming lemma (Lemma~\ref{renaming-lemma}). Since $\height(\Pi') = \height(\Pi)$ and the number of cut formulas is now $n-1$ it follows from induction hypothesis \ref{ind-cut-formulas} that $\Xi$ is reducible.

\refcase[red-6-2]{$-/$\axiom (\emph{Right-axiom case})}

If $\Pi$ ends with the $\axiom$ rule, then $n=1$,
$\Gamma = \emptyset$ and $C = B_1'$ for $B_1\approx B_1'$.
In this case, $\Xi$ reduces to $\Pi_1',$ where $\Pi_1'\sim\Pi_1$ is obtained from the renaming lemma (Lemma~\ref{renaming-lemma}). By the renaming congruence lemma (Lemma~\ref{renaming-congruence-lemma}), $\Pi'$ is reducible.

\end{proof}

\section{Additional properties of \LDIN}

\subsection{\DL for inductive definitions}
We make the following observation that inductive definitions are also fixed-point definitions:
\begin{lemma}[Inductive definitions are fixed-point definitions]
  \label{inductive-fixpoint}
  Given an inductive definition $p\ \vec x \ind B\ p\ \vec x$, the following rule is derivable
  \begin{mathpar}
    \inferrule*[right=\DL\!*]
    {\Gamma, B\ p\ \vec t \vdash C}
    {\Gamma, p\ \vec t \vdash C}
  \end{mathpar}
\end{lemma}
\begin{proof}
  The derivation is obtained from \muL with inductive invariant $S = B\ p$ as follows
  \begin{mathpar}
    \inferrule*[right=\muL]
    {B\ (B\ p)\ \vec x \vdash B\ p\ \vec x\\
      \Gamma, B\ p\ \vec t \vdash C}
    {\Gamma, p\ \vec t\vdash C}
  \end{mathpar}
  The derivability of $B\ (B\ p)\ \vec x\vdash B\ p\ \vec x$ is obtained from the following sublemma:
  \begin{addmargin}[1em]{1em}
    \begin{sublemma}
    Suppose $B$ is an operator of type $\omega\to\omega$ such that $p$ does not occur negatively in the formula $B\ p\ \vec x$, and suppose $S : \omega$ satisfies $S\ \vec t\vdash p\ \vec t$ for any $\vec t$, then $B\ S\ \vec x\vdash B\ p\ \vec x$ is derivable.
  \end{sublemma}
  \begin{proof}
    We proceed by induction on the complexity of $B\ p\ \vec x$. We only consider the following cases.
    \case{\underline{$p\ \vec t$}} In this case, $B\ S\ \vec x \equiv S\ \vec t$, so the derivation of $S\ \vec t\vdash p\ \vec t$ is obtained by assumption.

    \case{\underline{$B_1\ p\ \vec x \wedge B_2\ p\ \vec x$}} In this case, the derivation is obtained as follows
    \begin{mathpar}
      \inferrule*[right=\andL*]
      {\inferrule*[Right=\andR]
        {B_1\ S\ \vec x\vdash B_1\ p\ \vec x\\
          B_2\ S\ \vec x\vdash B_2\ p\ \vec x}
        {B_1\ S\ \vec x, B_2\ S\ \vec x\vdash B_1\ p\ \vec x\wedge B_2\ p\ \vec x}}
      {B_1\ S\ \vec x\wedge B_2\ S\ \vec x\vdash B_1\ p\ \vec x\wedge B_2\ p\ \vec x}
    \end{mathpar}
    where $B_i\ S\ \vec x\vdash B_i\ p\ \vec x$ are obtained from the induction hypotheses.
    
    \case{\underline{$B_1\ \vec x\imp B_2\ p\ \vec x$}} In this case, the derivation is obtained as follows
    \begin{mathpar}
      \inferrule*[right=\impR]
      {\inferrule*[Right=\impL, rightskip=3em, leftskip=3em]
        {\inferrule*[right=\init]{ }{B_1\ \vec x \vdash B_1\ \vec x}\\
          \inferrule*[Right=\wL]{B_2\ S\ \vec x\vdash B_2\ p\ \vec x}
          {B_2\ S\ \vec x,B_1\ \vec x\vdash B_2\ p\ \vec x}}
        {B_1\ \vec x\imp B_2\ S\ \vec x, B_1\ \vec x\vdash B_2\ p\ \vec x}}
      {B_1\ \vec x\imp B_2\ S\ \vec x\vdash B_1\ \vec x\imp B_2\ p\ \vec x}
    \end{mathpar}
    where $B_2\ S\ \vec x\vdash B_2\ p\ \vec x$ is obtained from the induction hypothesis.
    \case{\underline{$\forall_\alpha y.B'\ y\ p\ \vec x$}} In this case, the derivation is obtained as follows
    \begin{mathpar}
      \inferrule*[right=\allR]
      {\left\{\raisebox{-1em}{$\inferrule*[right=\allL]
        {B'\ t\ S\ \vec x\vdash B'\ t\ p\ \vec x}
        {\forall y. B'\ y\ S\ \vec x\vdash B'\ t\ p\ \vec x}$
        }\right\}_{t\in\ground(\alpha)}}
      { \forall y. B'\ y\ S\ \vec x\vdash \forall y. B'\ y\ p\ \vec x}
    \end{mathpar}
    where $(B'\ t)\ S\ \vec x\vdash (B'\ t)\ p\ \vec x$ is obtained from the induction hypothesis, noting that substituting $t$ for $y$ in $(B'\ t)$ does not increase the complexity of the formula.
  \end{proof}
\end{addmargin}

We apply the lemma by choosing $B = B$, and $S = (B\ p)$. Note that the stratification condition on inductive definitions guarantees that $p$ does not occur negatively in $B\ p$. Furthermore, the condition on $S$ is obtained as follows
\begin{mathpar}
  \inferrule*[right=\muR]
  {\inferrule*[Right=\init]{ }
    {B\ p\ \vec t\vdash B\ p\ \vec t}}
  {B\ p\ \vec t\vdash p\ \vec t}
\end{mathpar}

\end{proof}

\subsection{Uniqueness of definition instances}
We make the following observation from \cite{Tiu12} which makes use of the requirement that the head $H$ of a definition $H\defn B$ be in the pattern fragment.

\begin{lemma}[Uniqueness of definition instances]
  \label{definition-uniqueness-lemma}
  Given a definition $H\defn_{\mc X} B$ satisfying $\dfn(H\defn B, A, \theta, B_1)$ and $\dfn(H\defn B, A, \theta, B_2)$ such that $H\rho_1 = A\theta = H\rho_2$, then $\rho_1 = \rho_2$. In particular this implies that $B_1 = B_2$.
\end{lemma}
\begin{proof}
  We first recall the domains $\theta : \mc Z\to\mc Y$ and $\rho_1,\rho_2 : \mc Z\to\mc X$. We can then consider the problem of finding substitution $\rho : \mc Z\to \mc X$ unifying $A\theta$ and $H$. Since $\mc X\cap \mc Z = \emptyset$, and $H$ is in the pattern fragment, we know that their exists a unique most general unifier $\rho$ (see \cite{MN12}).
  In particular, since $\rho$ is a ground substitution with respect to $\mc X$, it is the unique solution and thus $\rho = \rho_1 = \rho_2$. Finally, we note that $B_1 = B\rho_1 = B\rho_2 = B_2$, as desired.
\end{proof}
}


\end{document}